\documentclass[sn-basic]{sn-jnl}

\usepackage{multirow}%
\usepackage{amsfonts}%
\usepackage{mathrsfs}%
\usepackage[title]{appendix}%
\usepackage{xcolor}%
\usepackage{textcomp}%
\usepackage{manyfoot}%
\usepackage{booktabs}%
\usepackage{algorithm}%
\usepackage{algorithmicx}%
\usepackage{algpseudocode}%

\usepackage{scalerel}
\usepackage{orcidlink}

\usepackage{amssymb}
\usepackage{amsmath}

\usepackage{amsthm}
\usepackage{proof}
\usepackage{stmaryrd}

\usepackage{booktabs}
\usepackage{dsfont}

\usepackage{bbm}
\usepackage{graphicx}
\usepackage{listings}
\usepackage{bussproofs}
\usepackage{tikz}
\usepackage{circuitikz}
\usepackage{floatflt} 

\usepackage{ifthen}
\usepackage{xspace}
\usepackage{pdfpages}
\usepackage{todonotes}
\usepackage{afterpage}
\usepackage{paralist}
\usepackage{xspace}
\usepackage{url}

\usepackage{pgf-umlsd}
\usetikzlibrary{arrows,shapes,trees,positioning,decorations,shapes}
\usetikzlibrary{decorations.markings,automata}

\usepackage[dash,dot]{dashundergaps}
\DeclareOldFontCommand{\rm}{\normalfont\rmfamily}{\mathrm}
\DeclareOldFontCommand{\sf}{\normalfont\sffamily}{\mathsf}
\DeclareOldFontCommand{\tt}{\normalfont\ttfamily}{\mathtt}
\DeclareOldFontCommand{\bf}{\normalfont\bfseries}{\mathbf}
\DeclareOldFontCommand{\it}{\normalfont\itshape}{\mathit}
\DeclareOldFontCommand{\sl}{\normalfont\slshape}{\@nomath\sl}
\DeclareOldFontCommand{\sc}{\normalfont\scshape}{\@nomath\sc}
\DeclareRobustCommand*\cal{\@fontswitch\relax\mathcal}
\DeclareRobustCommand*\mit{\@fontswitch\relax\mathnormal}

\colorlet{keywordcolor}{blue!50!black}
\colorlet{commentcolor}{green!60!black}
\colorlet{typecolor}{violet}
\newcommand{\sourcefont}{\ttfamily\small}
\newcommand{\commentfont}{\slshape\rmfamily\color{commentcolor}}

\lstdefinelanguage{ABS}{
        keywords={physical,duration,diff,differential,do,assert,this,new,data,type,def,case,of,local,class,interface,
        extends,implements,if,then,else,await,get,Fut,return,skip,while,module,
        import,export,from,suspend,adds,modifies,removes,original,productline,
        features,core,corefeatures,optionalfeatures,after,when,product,hasAttribute,
        hasMethod,hasField,hasInterface,uses,root,extension,group,allof,oneof,require,
        stateupdate,object,main,objectupdate,classupdate,fi,
        exclude,original,ifin,ifout,opt,null,
        newgroup,data,thiscomp,in,joins,leaves,subtypeOf,wait,acquire,except,as,component,Pre,Abs
        },
        keywordstyle=\color{keywordcolor}\bfseries\sffamily,
        morekeywords=[2]{Unit, Int, Bool, Rat, List, Set, Pair, Fut, Maybe, String, Triple, Either, Map, Real},
        keywordstyle=[2]\color{typecolor},
        sensitive=true,
        comment=[l]{//},
        morecomment=[s]{/*}{*/},
        morestring=[b]"
}

\lstdefinelanguage[v9]{Java}[]{Java}{
        morekeywords={module,requires,provides,uses,with,to,exports}
}
\lstdefinelanguage[ContextJ]{Java}[]{Java}{
        morekeywords={layer,with,without,proceed,before,after}
}
\lstdefinelanguage[FOP]{Java}[]{Java}{
        morekeywords={refines,original,Super}
}
\lstdefinelanguage[JastAdd]{Java}[]{Java}{
        morekeywords={aspect,syn,inh,lazy}
}

\lstdefinestyle{code}{
        basicstyle=\sourcefont\upshape,
        keywordstyle=\color{keywordcolor}\bfseries\sffamily,
        commentstyle=\commentfont,
        columns=fullflexible,
        mathescape=true,
        escapechar={\#},
        keepspaces=true,
        showstringspaces=false,
        aboveskip=8pt, 
        numbers=left,
        stepnumber=1, 
        numberstyle=\ttfamily\scriptsize\color{gray},
        numbersep=4pt,
        xleftmargin=1.5em,
        xrightmargin=1.5em,
        framexleftmargin=1.2em,
        framexrightmargin=1em,
        framextopmargin=0.5ex,
        breaklines=true,
        breakindent=3pt,
}

\lstdefinestyle{abs}{
        style=code,
        language=ABS,
}
\lstdefinestyle{java}{
        style=code,
            language=Java
}
\lstdefinestyle{java9}{
        style=code,
            language=[v9]Java
}
\lstdefinestyle{aspectj}{
        style=code,
        language=[AspectJ]Java
}
\lstdefinestyle{jastadd}{
        style=code,
        language=[JastAdd]Java
}
\lstdefinestyle{contextj}{
        style=code,
        language=[ContextJ]Java
}
\lstdefinestyle{FOP}{
        style=code,
        language=[FOP]Java
}
\lstdefinestyle{scala}{
        style=code,
        language=Scala,
        morekeywords={self}
}

\newcommand{\code}[2][]{\lstinline[style=code,basicstyle=\ttfamily\upshape,#1]{#2}}

\newcommand{\abs}[2][]{\code[style=abs,#1]{#2}}

\lstnewenvironment{srccode}[1][]{
        \minipage{\linewidth}
        \lstset{style=code,
        backgroundcolor=\color{white},
        rulecolor=\color{gray!50},
        frame=tblr,
        captionpos=b,
        #1}
}{\endminipage}

\lstnewenvironment{abscode}[1][]{
        \minipage{1\linewidth}
        \lstset{style=abs,
        backgroundcolor=\color{white},
        rulecolor=\color{gray!50},
        frame=tblr,
        captionpos=b,
        #1}
}{\endminipage}

\lstnewenvironment{javacode}[1][]{
        \minipage{\linewidth}
        \lstset{style=java,
        backgroundcolor=\color{gray!5},
        rulecolor=\color{gray!50},
        frame=tblr,
        captionpos=b,
        #1}
}{\endminipage}

\lstnewenvironment{jastaddcode}[1][]{
        \minipage{\linewidth}
        \lstset{style=jastadd,
        backgroundcolor=\color{gray!5},
        rulecolor=\color{gray!50},
        frame=tblr,
        captionpos=b,
        #1}
}{\endminipage}

 \makeatletter
\newcommand\BeraMonottfamily{%
  \def\fvm@Scale{0.85}
  \fontfamily{fvm}\selectfont
}
\makeatother

\lstdefinelanguage{KeYmaeraX}{%
  keywords={if,then,else,R,B,HP,Functions,ProgramVariables,Problem,End,Definitions,ArchiveEntry,Tactic,SharedDefinitions},%
  sensitive=true,
  morecomment=[s]{/*}{*/},
  deletestring=[d]',
  showstringspaces=false,
  commentstyle=\color{green},
  mathescape,
  escapeinside={/*@}{@*/}}[keywords]

\lstdefinelanguage{Bellerophon}{%
  language={},
  keywords={'R,'L,'_},%
  otherkeywords={;,<,|},
  sensitive=true,
  morecomment=[l]{//},
  morecomment=[s]{/*}{*/},
  morestring=[b]",
  deletestring=[d]',
  morestring=[d]`,
  showstringspaces=false,
  commentstyle=\fontseries{lc}\color{green}}[keywords]
  
\newcommand{\keycode}[1]{
    { \lstset{language=KeYmaeraX}
    \begin{lstlisting}
     #1
    \end{lstlisting}
    }
}

\newcommand{\COMMENT}[1]{}
\newcommand{\AUTOMATA}[1]{}

\newcommand{\HABS}{\ensuremath{\mathtt{HABS}}\xspace}

\DeclareMathOperator*{\argmin}{\mathbf{argmin}}

\let\temp\phi
\let\phi\varphi
\let\varphi\temp

\makeatletter
\newcommand{\xRightarrow}[2][]{\ext@arrow 0359\Rightarrowfill@{#1}{#2}}

\newcommand{\cased}[1]{\ensuremath{
\left\{\begin{array}{ll}
#1
\end{array}\right.
}\xspace}

\newcommand{\xabs}[1]{\text{\abs{#1}}}

\newcommand{\sem}[1]{\ensuremath{ \llbracket #1 \rrbracket \xspace}}
\newcommand{\ssep}{  \ | \ }

\newcommand{\many}[1]{\overline{#1}}
\newcommand{\TINFER}[3]{\begin{array}{c}\rulename{#1}  \frac{\begin{array}{c}#2
\\[0.5mm]
\end{array}
}{
\begin{array}{c}
\\[-3.5mm]
\displaystyle{#3}
\end{array}
}\end{array}}

\newcommand{\TTINFER}[4]{\begin{array}{c}\rulename{#1}  \frac{\begin{array}{c}#2
\\[1mm]
#3\\[0.5mm]
\end{array}
}{
\begin{array}{c}
\\[-3.5mm]
\displaystyle{#4}
\end{array}
}\end{array}}

\newcommand{\rulename}[1]{\textbf{\scriptsize(\textsf{#1})}}

\newcommand{\ddl}{\ensuremath{d\mathcal{L}}\xspace}
\newcommand{\dL}{\ddl}

\newcommand{\trace}{\ensuremath{\theta}\xspace}
\newcommand{\methodname}{\xabs{m}\xspace}
\newcommand{\classname}{\xabs{C}\xspace}
\newcommand{\statement}{\xabs{s}\xspace}

  \newcommand{\fid}{\ensuremath{\mathit{fid}}\xspace}

\usepackage{hyperref}
\usepackage{mathtools}

\newtheorem{definition}{Definition}%
\newtheorem{example}{Example}%
\newtheorem{lemma}{Lemma}%
\newtheorem{theorem}{Theorem}%
\newtheorem{proposition}{Proposition}%

\begin{document}

\title[Modular Analysis of Distributed Hybrid Systems using Post-Regions]{Modular Analysis of Distributed Hybrid Systems using Post-Regions (Full Version)}
\author{\fnm{Eduard} \sur{Kamburjan}}
\email{eduard@ifi.uio.no}

\affil{\orgdiv{Department of Informatics}, \orgname{University of Oslo},  \orgaddress{\city{Oslo}, \country{Norway}}}

\abstract{
We introduce a new approach to analyze distributed hybrid systems by a generalization of rely-guarantee reasoning.
First, we give a system for deductive verification of class invariants and method contracts in object-oriented distributed hybrid systems.
In a hybrid setting, the object invariant must not only be the post-condition of a method, but also has to hold in the \emph{post-region} of a method. 
The post-region describes all reachable states \emph{after} method termination \emph{before} another process is guaranteed to run. 
The system naturally generalizes rely-guarantee reasoning of discrete object-oriented languages to hybrid systems 
and carries over its modularity to hybrid systems: 
only one $d\mathcal{L}$-proof obligation is generated per method.
The post-region can be approximated using lightweight analyses and we give a general notion of soundness
for such analyses.
Post-region based verification is implemented for the Hybrid Active Object language \HABS. 
}


\maketitle

\section{Introduction}\label{sec:intro}
Distributed cyber-physical systems are behind modern innovation drivers such as the Internet-of-Things or Digital Twins.
Support for formal verification of such systems, however, is lagging behind their adaptation.
Especially \emph{deductive} verification of hybrid systems~\citep{Platzer10,QDL}, which does not suffer from the state-space explosion problem of model checking, 
has little compositionality mechanisms. 
This is in stark contrast to the situation for distributed \emph{discrete} systems, where a rich theory of composition principles for object-oriented languages exists,
e.g., method contracts~\citep{key,meyer} and object invariants.
Here, we show that this theory can be carried over to a hybrid setting.

\paragraph{Challenge.}
As a starting point, consider a system consisting of 
discrete, atomic procedures $\statement_i$, each guarded by a condition $\mathsf{grd}_i$. 
A procedure may run whenever its guard holds.
In this, still discrete, but already concurrent setting, to prove that some invariant $\mathsf{inv}$ holds, one can apply modular reasoning as follows:
One shows that the invariant holds in the initial state, and that it is preserved by each procedure;
Every procedure may \emph{assume} the invariant when it starts, but must also \emph{guarantee} it to other procedures.
This results in proof obligations of the following form for all procedures.
\[\mathsf{grd_i} \wedge \mathsf{inv} \rightarrow [\statement_i]\mathsf{inv}.\]
This dynamic logic~\citep{key,dynamic} formula expresses that if the invariant $\mathsf{inv}$ holds before executing procedure $\statement_i$, then it also holds after its execution.
If this formula is valid for all procedures, every procedure can indeed assume the invariant: 
When procedure $\statement_i$ runs after procedure $\statement_j$ has terminated, 
the state is the same as when procedure $\statement_j$ terminated, for which the invariant has been proven to hold.

Now, let us move to a hybrid setting and add some physical dynamics $\mathsf{dyn}$ to the system, that evolve the state according
whenever time advances. In this case the above proof obligations \emph{do not suffice:}
After a procedure terminates and before another procedure runs, the state may now evolve and the invariant does not necessarily hold.

\paragraph{Post-Regions.}
The main insight of this work is that the post-condition is not just the invariant $\mathsf{inv}$, but $\mathsf{inv}$ and a term that expresses that $\mathsf{inv}$ holds \emph{until the next procedure runs}. 
Given the guards of all procedures, we can compute the \emph{post-region} $\mathsf{pst}$ of a procedure $\statement_i$: A region where time may pass after $\statement_i$
terminates, because no procedure is \emph{guaranteed} to run. 
The post-condition is then a term expressing that the invariant holds inside the post-region $\mathsf{pst}$ \emph{while following the continuous dynamics} $\mathsf{dyn}$.
Such constraints are expressed with the formula $[\mathsf{dyn}\&\mathsf{pst}]\mathsf{inv}$ in differential dynamic logic~\citep{DBLP:conf/lics/Platzer12b} ($d\mathcal{L}$). 
In the most simple case, the proof obligation is the following:
\begin{align}\label{eq:intromain}
\mathsf{grd}_i \wedge \mathsf{inv} \rightarrow \big[\statement_i\big]\big(\mathsf{inv}\wedge\big[\mathsf{dyn}\&\mathsf{pst}\big]\mathsf{inv}\big).
\end{align}

\noindent In the trivial case, $\mathsf{pst}$ is $\mathsf{true}$, i.e., the invariant has to hold forever.
Using the guards of the known procedures, or other structure possibly provided by the used programming language, enables us to compute 
more precise post-regions.

In this article we use post-regions to verify safety properties for a more complex and structured model for distributed hybrid systems: Hybrid Active Objects.

\paragraph{Post-Regions for Hybrid Active Objects.}
\emph{Hybrid Active Objects} (HAO)~\citep{arxiv,lites} are a object-oriented formalism that combines the formal semantics and verification tools of low-level formalisms with the usability of mature programming languages. 
An HAO encapsulates continuous dynamics inside an object and defines methods to interact with it from the outside.
HAOs extend the Active Object concurrency model~\citep{BoerSHHRDJSKFY17}: message passing concurrency with futures, as well as cooperative scheduling.
With cooperative scheduling, only one process is active at each point in time (per object) and cannot be interrupted by other processes, unless it explicitly releases control over the object.
Active Objects have been applied in industrial case studies in many domains~\citep{AlbertBHJSTW14,BezirgiannisBJP19,LinLYJ20}. 
The biggest such case study with the Active Object language \texttt{ABS}~\citep{JohnsenHSSS10}, a model of railway operations~\citep{scp}, verification has shown that the
concurrency model is suitable for modular deductive verification~\citep{KamburjanH17}.

Post-regions generalize modular deductive verification to a hybrid setting, but the more complex concurrency model needs some adjustment in the proof obligations:
With cooperative scheduling, methods are not atomic and, thus, must also guarantee the invariant not only upon termination, but at every suspension point.
Furthermore, not every method can always run (as in the procedure-based system sketched above), but a single HAO can be called from other objects, or from within.
Thus, it must be verified in every possible context, and the internal calls can provide additional structure to compute post-regions.
Finally, several processes executing the same method can be active at the same time, and a HAO can deadlock due to its use of futures~\citep{DBLP:journals/toplas/Halstead85,DBLP:journals/jfp/FlanaganF99}, which must be taken care of in the verification process as well.

\paragraph{Contributions.}
Our main contributions are (1) \emph{post-regions}, a natural generalization of assume-guarantee reasoning from discrete programming languages to a hybrid object-oriented language, and
(2) \emph{a general notion of soundness} for analyses that derive post-regions.
The system is designed with a focus on \emph{modularity}: one proof obligation is used per method, and after changes to one method it is not necessarily required to reprove the whole program.

\paragraph{Structure.}
Sec.~\ref{sec:logic} introduces the $d\mathcal{L}$ logic, on which our verification builds.
Sec.~\ref{sec:concurrent} then introduces post-regions for a simple concurrency model to illustrate the general structure.
Sec.~\ref{sec:habs} describes the \HABS hybrid active object language and 
Sec.~\ref{sec:region} gives the verification system for it. 
Sec.~\ref{sec:multi} discusses practical issues and describes the implementation.
Sec.~\ref{sec:related} gives related work and Sec.~\ref{sec:conclusion} concludes.

\paragraph{Origin of Material.}
This manuscript is an extended and reworked version of~\citep{HSCC}:
Sections \ref{sec:concurrent} and \ref{ssec:general} are completely new, with new semantic foundations and notion of soundness for post-regions.
The presentation of Section \ref{sec:region} has been adjusted to this new notion of soundness and post-region generators, and the related work has been extended.

\section{Preliminaries: Differential Dynamic Logic}\label{sec:logic}
Differential dynamic logic (\ddl)~\citep{QDL,Platzer18} is a first-order dynamic logic, that embeds hybrid programs into its modalities.
It is implemented in the \texttt{KeYmaera X} tool~\citep{FultonMQVP15}
Hybrid programs are defined by a simple while-language, extended with a statement for ordinary differential equations (ODEs). 
Such a statement evolves the state according to some dynamics for a non-deterministically chosen amount of time.

\begin{definition}[Syntax of $d\mathcal{L}$]
Let $p$ range over predicate symbols (such as $\doteq,\geq$), $f$ over function symbols (such as $+,-$) and \abs{x} over variables.
Hybrid programs $\alpha$, formulas $\phi$ and terms $t$ are defined by the following grammar.
\begin{align*}
\phi ::=~&p(\many{t}) \ssep t \doteq t \ssep \neg\phi \ssep \phi \wedge \phi \ssep \exists \xabs{x}.~\phi \ssep [\alpha]\phi\qquad t ::=~f(\many{t}) \ssep \xabs{x} \quad \mathit{dt} := f(\many{dt}) \ssep t \ssep (t)'\\
\alpha ::=~&\xabs{x :=}~t \ssep \xabs{x := *} \ssep \alpha \cup \alpha \ssep \alpha^\ast \ssep ?\phi \ssep \alpha;\alpha \ssep \{\alpha\} \ssep \many{\xabs{x =}~\mathit{dt}} \& \phi 
\end{align*}
\end{definition}
Modalities may be nested using the $?$ operator and all ODEs are autonomous.
The semantics of hybrid programs is as follows:
Program $\xabs{x :=}~t$ assigns the value of $t$ to \abs{x}. 
Program $\xabs{x := *}$ assigns a non-deterministically chosen value to \abs{x}. 
Program $\alpha_1 \cup \alpha_2$ is a non-deterministic choice.
Program $\alpha^\ast$ is the Kleene star.
Program $?\phi$ is a test or filter. It either discards a run (if $\phi$ does not hold) or performs no action (if $\phi$ does hold).
Program $\alpha_1;\alpha_2$ is sequence and $\{\alpha\}$ is a block for structuring. Finally, the ODE $\many{\xabs{x =}~\mathit{dt}} \& \phi$
evolves the state according to the given ODE in evolution domain $\phi$ for some amount of time. 
The evolution domain describes where a solution is allowed to evolve, not the solution itself.
The semantics of the first-order fragment is completely standard.
The semantics of $[\alpha]\phi$ is that $\phi$ has to hold in \emph{every} post-state of $\alpha$ if $\alpha$ terminates.
We stress that if $\alpha$ is an ODE, then this means that $\phi$ holds throughout the \emph{whole} solution.
\begin{example}
The following formula expresses that the dynamics of bouncing ball with position \abs{x} and velocity \abs{v} are below 10 before the ball reaches the ground, when starting with a velocity of 0 and a height below 10.
\[0 \leq \xabs{x} \leq 10 \wedge \xabs{v} \doteq 0 \rightarrow [\xabs{x}' = \xabs{v}, \xabs{v}' = -9.81 \& \mathtt{x} \geq 0]\xabs{x} \leq 10\]
Events can be expressed as usual by an event boundary created between a test and an evolution domain. 
The following program models that the ball repeatedly bounces back exactly on the ground.
\[\big(\{\xabs{x}' = \xabs{v}, \xabs{v}' = -9.81 \& \mathtt{x} \geq 0\}; ? \mathtt{x} \leq 0; \xabs{v} := -\xabs{v}*0.9\big)^\ast\]
\end{example}
\ddl formulas are evaluated over interpretations, which assign predicate and function symbols to predicates and functions, as well as valuations, which map program variables to value.
Standard control flow (such as $\mathsf{while}$ and $\mathsf{if}$) is encoded using the operators above~\citep{DBLP:conf/lics/Platzer12b}. 
We denote the states reachable from a valuation $\sigma$ by program $\alpha$ with $\llbracket\alpha\rrbracket_\sigma$.
If $\alpha$ is an ODE, we denote the states reachable from $\sigma$ by following the semantics for exactly $t$ time units with $\llbracket\alpha\rrbracket_{\sigma,t}$.
The formal semantics, interpretation of predicate and function symbols is independent of our languages and also given by \cite{DBLP:conf/lics/Platzer12b}.

We use weak negation to preserve event boundaries. It is defined as normal negation, except for weak inequalities:
\[(\widetilde{\neg} t_1 \leq t_2) = t_1 \geq t_2 \qquad (\widetilde{\neg} t_1 \geq t_2) = t_1 \leq t_2\]
Given a formula $\phi$, its conjuction with its weak negation defines a boundary around the region defined by $\phi$: 
Given a trajectory starting at a point in $\neg \phi$ and ending at a point in $\phi$, the first point of the trajectory in $\phi$
is in $\phi \wedge \widetilde{\neg} \phi$ (if $\phi$ only contains weak inequalities).

\section{Region-Based Verification}\label{sec:concurrent}
We introduce the notion of post-regions first for the simple concurrency model sketched in the introduction to illustrate its core ideas, 
before we move on to the more complex concurrency model for distributed systems using Hybrid Active Objects.

\paragraph{Concurrent Programs.}
The used concurrency model uses guarded procedures to react on changes in the global state
The global state evolves according to some continuous dynamics whenever no procedure is active and no guard holds.
Each procedure is running atomically (i.e., no interleaving or preemption) and can be executed arbitrarily often, provided that its guard holds.
There is communication between the procedures except the (shared) variables in the global state and the complete system is known.
The verification task will be to ensure that every reachable state is safe w.r.t.\ some formula, while analyzing each procedure in isolation.
This will suffice to illustrate the main ideas behind post-regions, in particular the trade-off between precision and robustness.

Formally, \emph{a concurrent program} is a set of guarded \ddl programs without ODEs, and one ODE for continuous dynamics.
\begin{definition}[Concurrent Program]
Let \emph{guards} $\mathsf{grd}$ range over \ddl formulas of the form $\bigwedge t \leq t'$.
Let \emph{procedures} $\mathsf{prcd}$ range over \ddl programs of the form $?\mathsf{grd}; \mathsf{bdy}$, where $\mathsf{bdy}$ contains no ODEs.
Let $\mathsf{dyn}$ range over ODEs without evolution domain.
A \emph{concurrent program} has the form $((\mathsf{prcd}_i)_{i\in I}, \mathsf{dyn})$.
\end{definition}
To access the guard of a procedure $\mathsf{prcd}$, we write $\mathsf{grd}_\mathsf{prcd}$.
For states, we use \ddl valuations, together with a clock to keep track of time.
\begin{definition}[Concurrent State]
A concurrent state $t, \sigma$ is a pair of a valuation $\sigma$, which maps program variables to real numbers, and a clock $t \in \mathbb{R}^+$.
\end{definition}

The semantics of a program is defined by a transition system between concurrent states.
It consists of two rules:
The first rule executes a procedure, if the guard of the procedure holds.
The second rule advances time until the guard of some procedure holds if no procedure can be executed right now.
This realizes urgent transitions defined by the procedures and is well-defined as we force the guards to use non-strict inequalities.
\begin{definition}[Concurrent Transition System]
The transition system for concurrent programs is defined by the two rules in Fig.~\ref{fig:sem:conc:rules}, 
which are parametric in the program $(B, \mathsf{dyn})$.
\begin{figure}[!b]
\[
\TINFER{execute}{
\mathsf{prcd} \in B \qquad 
\mathsf{prcd} = ?\mathsf{grd};\mathsf{bdy}\qquad
\sigma \models \mathsf{grd}\qquad
\sigma' \in \llbracket \mathsf{bdy} \rrbracket_\sigma
}{
 t, \sigma \xrightarrow{B, \mathsf{dyn}} t, \sigma'
}
\]
\[
\TTINFER{urgent}{ 
    \forall \mathsf{prcd} \in B.~\sigma \not\models \mathsf{grd}_\mathsf{prcd}
}{
    t' = \argmin_{t' > 0} \exists \mathsf{prcd} \in B. \sigma' \in \llbracket \mathsf{dyn}\&\mathsf{true}\rrbracket_{\sigma, t' - t} \wedge \sigma' \models \mathsf{grd}_\mathsf{prcd}
}{
 t, \sigma \xrightarrow{B, \mathsf{dyn}} t', \sigma'
}
\]
\caption{Runtime Semantics of Concurrent Programs.}
\label{fig:sem:conc:rules}
\end{figure}
\end{definition}
The rule \rulename{execute} picks some procedure $\mathsf{prcd}$, checks whether its guard holds and, if this is the case, picks some final state from the evaluation
of the contained \ddl-program. If the program is non-deterministic, then so is the overall system.
The rule \rulename{urgent} has two premises. The first checks that rule \rulename{execute} cannot fire, the
second one computes the minimal $t'$ such that advancing the time by $t'-t$ time units results in a state where some guard holds.

\begin{example}\label{ex:tank:prog}
The classical water tank example can be expressed with the following program.
The system keeps the water level between 3 and 10.
\begin{align*}
\mathsf{dyn} &= \{ \mathtt{level}' = \mathtt{drain}\}\\
\mathsf{prcd}_\mathsf{up} &= ?\mathtt{level} \geq 10; \mathtt{drain} = -1\\
\mathsf{prcd}_\mathsf{down} &= ?\mathtt{level} \leq 3; \mathtt{drain} = 1
\end{align*}
\end{example}

It is worth stressing that we separate continuous dynamics and discrete computations -- their interleaving is handled by the runtime semantics and is not explicitly encoded.
Implicitly, all procedures run in parallel with the dynamics, but we do not treat this parallelism in its own programming construct. 

Next, we define reachability and safety.

\begin{definition}[Runs, Reachable States, Safety]
Given a valuation $\sigma_0$, and a program $(B,\mathsf{dyn})$, a run is defined as a (possibly infinite) reduction sequence
\[0, \sigma_0 \xrightarrow{B, \mathsf{dyn}} t_1, \sigma_1 \xrightarrow{B, \mathsf{dyn}} \dots\]

If the sequence is finite, then the last state $t, \sigma$ is \emph{final}.
For all concurrent states $t_i, \sigma_i$, we say that they are part of the run.
The set of reachable states of the program $(B,\mathsf{dyn})$ starting from $\sigma_0$ is defined as follows:
\begin{align*}
  &\{
    \sigma \ssep t, \sigma \text{ is part of some run for some $t$}
  \}\\
\cup&
  \{
    \sigma'' \ssep t, \sigma \xrightarrow{B, \mathsf{dyn}} t', \sigma' 
    \text{ is part of a run, and }\\
    &\hspace{9mm}\exists t''.~t<t''<t'\wedge\sigma'' \in \llbracket\mathsf{dyn}\rrbracket_{\sigma, t'-t''}
  \}\\
  \cup&
  \{
    \sigma' \ssep t, \sigma \text{ is the final concurrent state and $\sigma' \in \llbracket\mathsf{dyn}\rrbracket_{\sigma}$}
  \}
\end{align*}

\noindent A program is \emph{safe} w.r.t.\ a \ddl-formula $\mathsf{inv}$, if all its reachable states are models for $\mathsf{inv}$.
\end{definition}
The first set are those states where a procedure starts or ends, the second computes all the states between two procedures,
which are only implicitly available in the run. 
The last set collects all the states which are reachable after the last procedure is executed, 
as these states are not explicitly generated by the transition systems.
Note that our notion of safety does not require the intermediate states of procedures to be safe, i.e., they are not strong invariants.


\paragraph{Post-Conditions.}
We aim for robust, modular, deductive verification of safety properties for unbounded systems, i.e., 
systems with infinite state spaces.
\begin{itemize}
\item The structure of the analysis should be modular, i.e., it decomposes the verification problem for the overall system to verification subproblems of subsystems.
\item The analysis should be robust w.r.t.\ changes, this means that changes in one subsystem should influence as few verification subproblems of other subsystems as possible. 
\end{itemize}

For reference, we give a simple system that handles the discrete fragment of our language, using standard rely-guarantee reasoning.
It checks that the invariant holds in the initial state, and that every method preserves it.
The invariant is as both pre- and post-condition for each procedure: it is assumed to hold when the procedure starts and must be guaranteed to hold in the state when it terminated.

We say that the \ddl-ODE $\mathsf{trv}$ is trivial if it is equivalent to the empty ODE that manipulates no variable. 
A system with a trivial ODE is not hybrid.
\begin{lemma}[Post-Conditions]\label{lem:hcp:1}
Let $\mathsf{inv}$ be a \ddl formula without modalities.
Let $(B,\mathsf{trv})$ be a program with trivial dynamics, and $\sigma$ a valuation.
Let $\mathsf{init}_\sigma \equiv \bigwedge_{\mathtt{v} \in \mathbf{dom}\sigma} \mathtt{v} \doteq \sigma(\mathtt{v})$ be the formula encoding $\sigma$ in \ddl.
If the \ddl-formula
\[ \mathsf{init}_\sigma \rightarrow \mathsf{inv} \]
is valid, and for every $?\mathsf{grd};\mathsf{bdy} \in B$ the formula
\[
\mathsf{inv} \rightarrow [?\mathsf{grd};\mathsf{bdy}]\mathsf{inv}
\]
is valid, then every reachable state is safe, i.e., a model for $\mathsf{inv}$.
We refer to the above formulas as \emph{proof obligations}.
\end{lemma}
\begin{proof}
We show that every reachable state within an arbitrary run is safe by showing that for any $n$, the run up to the $n$ transition is safe.
Induction base $n=0$: This requires that $\sigma$ is a model for $\mathsf{inv}$, which is exactly the semantics of the first proof obligation.
Induction step $n>0$: This means that the run is currently in a concurrent state $t, \sigma$ which is safe and makes a transition into another state
$t', \sigma'$. First, we observe that \rulename{urgent} never fires for trivial dynamics, so \rulename{execute} fires, $t = t'$ and $\sigma' \in \llbracket?\mathsf{grd};\mathsf{bdy}\rrbracket_\sigma$ for some procedure $?\mathsf{grd};\mathsf{bdy}$. The second proof obligation expresses exactly that every such state is a model for $\mathsf{inv}$.
\end{proof}

The system is (1) modular, as safety of the system is broken down to safety of procedures, (2) robust, as after changes in one procedure it is not required to reprove the proof obligation 
of other procedures\footnote{Changes in the safety invariant require to reprove everything.}, and (3) handles unbounded systems in the sense that every procedure can be executed arbitrarily often.
It is easy to see, however, that this proof obligation scheme does not suffice for non-trivial dynamics -- the rule \rulename{urgent} may actually fire and nothing guarantees that the state afterwards is safe.

To ensure safety for non-trivial dynamics, the state must remain safe through the evolution of the state until the next procedure starts.
A post-condition, however, only reasons about the immediate post-state, the state after procedure termination.
This brings us to \emph{post-regions}. The post-region of a procedure is the region in the state space where the invariant must continue to hold.
In the most simple case, this post-region is just $\mathsf{true}$ -- the system remains safe forever, even if no other procedure ever runs.

\begin{lemma}[Basic Post-Regions for Concurrent Hybrid Programs]\label{lem:hcp:2}
Let $(B,\mathsf{dyn})$ be a program, and $\sigma$ a valuation.
If the \ddl-formula
\[ \mathsf{init}_\sigma \rightarrow \mathsf{inv} \wedge [\mathsf{dyn} \& \mathsf{true}]\mathsf{inv}\]
is valid, and for every $?\mathsf{grd};\mathsf{bdy} \in B$ the formula
\[
\mathsf{inv} \rightarrow [?\mathsf{grd};\mathsf{bdy}](\mathsf{inv} \wedge [\mathsf{dyn} \& \mathsf{true}]\mathsf{inv})
\]
is valid, then every reachable state is safe, i.e., a model for $\mathsf{inv}$.
\end{lemma}
\begin{proof}
The proof is analogous to the one of Lemma~\ref{lem:hcp:1}, except for the case that the transition fires
due to rule \rulename{urgent}: 
The system is in a state $t, \sigma$, where
$\sigma$ is safe and the transition results in a state $t', \sigma'$.
We need to show that for every $t_s \in (0, t'-t]$, the states $\sigma' \in \llbracket \mathsf{dyn}\rrbracket_{\sigma,t_s}$ are safe.
The modified proof obligation, however is even stronger; it ensures that all states are safe:
\[
\bigcup_{t_s \in (0, t'-t]}\llbracket \mathsf{dyn}\rrbracket_{\sigma,t_s}
\subseteq \llbracket \mathsf{dyn}\rrbracket_{\sigma}
\]
The same holds for the states which are reachable after the last \rulename{execute} transition, if such a last transition exists.
\end{proof}
The above proof obligations can be simplified, but is chosen to mirror the syntactic structure of systems introduced later.
This proof obligation scheme is sound and inhibits the same properties as the one for discrete systems.
It is, however, not expressive enough. 
Even the above example cannot be verified with it -- for the procedure $\mathsf{prcd}_\mathsf{up}$, the following has to be proven (for an invariant $3 \leq \mathtt{level} \leq 10$):

\noindent\resizebox{\textwidth}{!}{
\begin{minipage}{1.05\textwidth}
\begin{align*}
&3 \leq \mathtt{level} \leq 10 \\
\rightarrow &\big[?\mathtt{level} \geq 10; \mathtt{drain} = -1\big]\big(3 \leq \mathtt{level} \leq 10 \wedge [\mathtt{level}' = \mathtt{drain} \& \mathsf{true}]3 \leq \mathtt{level} \leq 10\big)
\end{align*}
\end{minipage}
}
\vspace{2mm}

It is easy to see that this does not hold: if the dynamics continue forever, then the level will drop below 3 at some point.
The program, however, does not do so -- it follows the continuous dynamics only until another procedure runs.
This knowledge must be included in the proof obligations, and we do so by giving a more precise post-region:
the state must only stay safe as long as none of the guards of any procedure hold, as once any guard is guaranteed to hold, the system will \emph{immediately} react.

\begin{lemma}[Precise Post-Regions for Concurrent Hybrid Programs]\label{lem:hcp:3}
Let $(B,\mathsf{dyn})$ be a program, and $\sigma$ a valuation.
Let $\mathsf{pst} = \bigwedge_{\mathsf{prcd} \in B} \tilde\neg\mathsf{grd}_\mathsf{prcd}$ be the conjunction of all weakly negated guards.
If the \ddl-formula
\[ \mathsf{init}_\sigma \rightarrow \mathsf{inv} \wedge [\mathsf{dyn} \& \mathsf{pst}]\mathsf{inv} \]
is valid, and for every $?\mathsf{grd};\mathsf{bdy} \in B$ the formula
\[
\mathsf{inv} \rightarrow [?\mathsf{grd};\mathsf{bdy}](\mathsf{inv} \wedge [\mathsf{dyn} \& \mathsf{pst}]\mathsf{inv})
\]
is valid, then every reachable state is safe, i.e., a model for $\mathsf{inv}$.
\end{lemma}
\begin{proof}
The proof is analogous to the one of Lemma~\ref{lem:hcp:2} except for the case for rule \rulename{urgent}:
The system is in a state $t, \sigma$, where
$\sigma$ is safe and the transition results in a state $t', \sigma'$.
We need to show that for every $t_s \in (0, t'-t]$, the states $\sigma' \in \llbracket \mathsf{dyn}\rrbracket_{\sigma,t_s}$ are safe.
Here, let $\mathsf{prcd} = ?\mathsf{grd};\mathsf{bdy}$ be the procedure which fires \emph{next}. Note that
this procedure is guaranteed to exist, as otherwise \rulename{urgent} could not fire. 
We observe that except for the final state, the guard $\mathsf{grd}$ does not hold,
thus no state in the set $\bigcup_{t_s \in (0, t'-t)}\llbracket \mathsf{dyn}\&\mathsf{pst}\rrbracket_{\sigma,t_s}$ is a model for $\mathsf{grd}$.
We can, thus overapproximate this set with the evolution domain $\neg \mathsf{grd}$. The set $\bigcup_{t_s \in (0, t'-t]}\llbracket \mathsf{dyn}\&\mathsf{pst}\rrbracket_{\sigma,t_s}$
can be overapproximated by the formula $\tilde\neg \mathsf{pst}$.
This is exactly the condition implied by the modified proof obligation, which guarantees this for \emph{all} possible next procedures.

The same holds for the states which are reachable after the last procedure run, if such a last transition exists, as if it would run out of the post-region, then another procedure would fire.
\end{proof}
With this proof obligation scheme, the above example can be verified.
It is worth pointing out that it is not verified that the \emph{post-region} is safe. Instead, it is verified that the \emph{dynamics from every post-state} stay safe within the post-region.

Let us examine the post-condition $\mathsf{inv} \wedge [\mathsf{dyn} \& \mathsf{pst}]\mathsf{inv}$.
The first conjunct states that the invariant holds after the procedure terminates and the second says that it remains safe as long as the dynamics $\mathsf{dyn}$ remain in the post-region $\mathsf{pst}$. The first conjunct is not generally implied by the second -- if $\mathsf{pst} \equiv \text{false}$, then the second conjunct is trivially true.
This corresponds to the case where directly after termination of the first procedure another procedure runs.
As it assumes the invariant, we must verify it even if the state does not evolve.


The given programming model can be easily encoded in \ddl directly, and encoding urgency will use exactly the same weak negations as used for the post-region.
This is an artifact of the simplicity of the concurrency model -- only one process per procedure can run at the same time, and there is a single interleaving point (at procedure start/end). Furthermore, the system is closed -- the procedures cannot communicate with the outside.
After having presented the core ideas here, we use post-regions in the remaining section for a more 
complex concurrency model, where post-region based analysis cannot be reduced to a single \ddl formula, and that requires different patterns and structures in the programs to compute post-regions.

\section{Hybrid Active Objects}\label{sec:habs}
In this section we present the \emph{Hybrid Abstract Behavioral Specification} (\texttt{HABS}) language that implements Hybrid Active Objects.
For brevity's sake we refrain from introducing full \texttt{HABS}~\citep{arxiv,lites} and omit, e.g., inheritance, traits, product lines and field initialization.

Before formally introducing the language, we introduce the concurrency model with a discrete example, and we give an informal example to show the main ideas behind the hybrid extension.
For a detailed introduction to Active Objects, especially the \texttt{ABS} language, we refer to \cite{JohnsenHSSS10}, for a discussion in light of verification to \cite{doa}. 

\HABS is an object-oriented language, where each object is running on its own node in a distributed network and is \emph{strongly encapsulated}.
This means that an object cannot access the memory of another object, even if it is from the same class.
Internally, each object is implementing \emph{cooperative scheduling}: for each object, at most one process is active, i.e., can be executed and can access the object's memory.
\HABS is preemption-free: a process can only become inactive, if it \emph{explicitly} deschedules itself. A process is descheduled either by termination, or by special statements for suspension.
Between objects, communication is realized by asynchronous method calls and \emph{futures}: each asynchronous method call generates a future for the callee.
The future identifies the called process and is said to be \emph{resolved} once the associated called process terminates.
The callee can pass the future around and any process can synchronize on it, i.e., wait until it is resolved and a return value can be read.
Such synchronization is either blocking (the process remains active until the future is resolved) or suspending (the process deschedules and can only become active again once the future is resolved).

%
%
%
%
We present a water tank model to illustrate the language.
\begin{example}\label{ex:bball}
Fig.~\ref{fig:bball} shows the \HABS model of a water tank that keeps a water level between 3$l$ and 10$l$.
The pictured class, \abs{Tank} has one discrete field (\abs{log}) and
a \abs{physical} block that introduces physical fields. A physical field is described by its initial value and an ODE.
Here, \abs{level} and \abs{drain} model water level and drain of the tank. The drain is constant and the water level changes linearly w.r.t.\ the drain. 
Additionally, an initialization block is provided, which calls the methods 
\abs{up} and \abs{down}.
Each method starts with an \abs{await} statement, whose guard is the condition when the process will be scheduled (for \abs{up}, at the moment the level reaches 10 while water rises). 
Guards have an \emph{urgent} semantics and trigger as soon as possible.  
Each execution of a method is logged by calling the external object  \abs{log} on method \abs{triggered}.
Such method calls to other objects are asynchronous, i.e., execution continues without waiting for it to finish.
Then, the drain is adjusted and the method calls itself recursively to react the next time.
\end{example}
\begin{figure}
\noindent\begin{abscode}
class Tank(Log log){
  physical{ 
    Real level = 5; level' = drain;  //ODE and initial value for level
    Real drain = -1; drain' = 0;     //ODE and initial value for drain
  }
  { this!up(); this!down(); }         //Constructor
  Unit down(){
    await diff level <= 3 & drain <= 0;
    log!triggered(); drain =  1; this!down();
  }
  Unit up(){
    await diff level >= 10 & drain >= 0;
    log!triggered(); drain = -1; this!up();
  }
}
\end{abscode}
\caption{A water tank in \HABS.}
\label{fig:bball}
\end{figure}

The example provides us a first glimpse on post-regions: after a process terminates, the post-region
is the set of states reachable by the dynamics. 
In our case, both \abs{up} and \abs{down} are always potentially schedulable (i.e., there is a process that can be scheduled if the guard expression holds).
Thus, the post-region can be described by the weak negation of their guards: 
\[(\xabs{level} \geq 3 \vee \xabs{drain} \geq 0) \wedge (\xabs{level} \leq 10 \vee \xabs{drain} \leq 0)\]

Let us now look at an example with more complex interactions.
Consider the two methods in Fig.~\ref{fig:bball2}. 
Method \abs{slowDrain} halves the drain rate for 1 time unit, if the water level is above $3.5l$.
This is done with the \abs{duration(1)} statement. It advances time by 1 time unit, while \emph{blocking the object} -- no other process can execute, 
including any methods that are (asynchronously) called in this time. As time advances, the state evolves according to the dynamics.
Yet, the water level is obviously kept above 3$l$ -- the drain is always above -1, thus in one time unit with halved drain with $\xabs{level} \geq 3.5$, 
at most 0.5$l$ can leak.

Now consider the \abs{flushLog} method. Its first line retrieves the number of entries from the logger by calling \abs{getNumberEntries}.
This call generates a future, stored in variable \abs{f}, but is asynchronous -- execution continues after the call.
The next line uses \abs{get} to retrieve the return value of the call and blocks the object until the called method terminates.
This can make the system unsafe -- if the logger deadlocks, the method never terminates, thus the object remains blocked and the \abs{up} and \abs{down} methods cannot execute 
to keep the water level within limits.

\begin{figure}
\begin{abscode}
class Tank(Log log){
  //... as in Fig. 2
  Unit slowDrain(){
     if(level >= 3.5 && drain <= 0){
       drain = drain/2;
       duration(1);
       drain = drain*2;
     }
  }

  Unit flushLog(){
      Fut<Int> f = log!getNumberEntries();
      Int i = f.get;
      if(i >= 100) log!flush();
  }
}
\end{abscode}
\caption{Additional methods for the water tank in \ref{fig:bball}.}
\label{fig:bball2}
\end{figure}

An important aspect of the Active Object concurrency model is that each object is \emph{strongly} encapsulated. No other object can access the fields of an instance.
Inside an object, only one process is active at a time. This process cannot be preempted by the scheduler --- 
it must explicitly release control by either terminating or suspending (via \abs{await}). 
These two properties make (hybrid) active object easy to analyze: 
approaches for sequential program analyses can be applied
between two \abs{await} statements (and method start and method end).

The rest of this section is structured as follows. Sec.~\ref{sec:lang:syntax} gives the syntax of \HABS and an intuition for the runtime semantics.
It also illustrates further challenges for verification and introduces the \emph{controller} pattern.
Sec.~\ref{sec:lang:semantics} gives the transition system that defines the runtime semantics and Sec.~\ref{sec:lang:traces} describes how traces are extracted from the transition system.

\subsection{Syntax}\label{sec:lang:syntax}

The syntax of \HABS is given by the grammar in Fig.~\ref{fig:syntax}.
Program point identifier $\mathsf{p}$ range over natural numbers.
Standard expressions \abs{e} are defined over fields \abs{f}, variables \abs{v} and
operators \abs{!}, \abs{|}, \abs{\&}, \abs{>=}, \abs{<=}, \abs{+}, \abs{-},
\abs{*}, \abs{/}.  
Differential expressions \abs{de} are expressions extended with a derivation operator \abs{e'}.
Types \abs{T} are all class names \abs{C}, type-generic
futures \abs{Fut<T>}, \abs{Real}, \abs{Unit} and \abs{Bool}. 

\begin{figure}
\begin{align*}
\mathsf{Prgm} ::=~& \many{\mathsf{ClassDecl}}~\{\mathsf{s}\} && \text{\small Programs}\\
\mathsf{ClassDecl} ::=~& \xabs{class C}\left[\xabs{(}\many{\xabs{T f}}\xabs{)}\right]\big\{[\mathsf{Phys}]~\big[\{\mathsf{s}\}\big]~\many{\mathsf{MetDecl}}\big\} 
&& \text{\small Class Declarations}\\
\mathsf{MetDecl} ::=~& \xabs{T m(}\many{\xabs{T}~\xabs{v}}\xabs{)}~\{\mathsf{s}\xabs{; return e;}\}
&& \text{\small Method Declarations}\\
\mathsf{Phys} ::=~& \xabs{physical}~\{\many{\mathsf{PhysDecl}}\} && \text{\small Physical Block}\\
\mathsf{PhysDecl} ::=~& \xabs{Real f = e}:~\xabs{f' = de} 
&&\text{\small Physical Field Declaration}\\
\mathsf{s} ::=~&
\xabs{while (e)}~\{\mathsf{s}\}\ssep
\xabs{if (e)}~\{\mathsf{s}\}~[\xabs{else}~\{\mathsf{s}\}]\ssep
\mathsf{s}\xabs{;}\mathsf{s}  && \\
&\ssep\xabs{await}_{\mathsf{p}}~\mathsf{g} 
\ssep
[[\xabs{T}]~\xabs{e} =] \mathsf{rhs} \ssep \xabs{duration(e)}&&\text{\small Statements}\\
\mathsf{g} ::=~&\xabs{e?} \ssep \xabs{duration(e)} \ssep \xabs{diff e}
&&\text{\small Guards}\\
\mathsf{rhs} ::=~&
\xabs{e} \ssep \xabs{new C(}\many{\xabs{e}}\xabs{)} \ssep 
\xabs{e.get} \ssep
\xabs{e!m(}\many{\xabs{e}}\xabs{)}
&&\text{\small RHS Expressions}
\end{align*}
\caption{\HABS grammar. Notation $[\cdot]$ denotes optional elements and $\many{~\cdot~}$ lists.}
\label{fig:syntax}
\end{figure}

A program consists of a set of classes and a main block. Each class may have a list of discrete fields that are passed as parameters on object creation, 
an optional physical block, a list of discrete fields that are not passed as parameters, a set of methods and an initializing block that is executed after object creation.

The \abs{physical} block is a list of field declarations followed by an ODE describing the dynamics of the declared fields. The fields declared inside a \abs{physical} block are called \emph{physical}.
Methods, initializing blocks and the main block consist of statements. Non-standard statements, compared to a \texttt{while} language, are the asynchronous method calls and the following statements, for which we give an intuition here before formally defining them.
\begin{itemize}
\item The \abs{duration(e)} statement advances time by \abs{e} time units. No other process on the same object may execute during this advance.
Time in \HABS is global and symbolic. 
\item The \abs{e.get} statement reads from a future. A future is a container that is generated by an asynchronous call (\abs{o!m()}). 
Afterwards a future may be passed around. With the \abs{get} statement one can read the return values once the called process terminates. Until then, the reading process blocks and no other process can run on the object.
\item The $\xabs{await}_{\mathsf{p}}~\xabs{g}$ statement suspends the process until the guard \abs{g} holds. 
A guard is either (1) a future poll \abs{e?} that determines whether the call for the future in \abs{e} has terminated, (2) a duration guard that advances time while \emph{unblocking} the object (i.e., while time advances other processes may run on this object), or (3) a differential guard \abs{diff e} that holds once expression \abs{e} holds.
The $\mathsf{p}$ is a program-point identifier that uniquely identifies an \abs{await} statement in a program. 
Program-points identifiers are an auxiliary structure that we use later during proof obligation generation. We omit them whenever they are not needed.
\end{itemize}

These statements pose additional challenges to verification of invariants, as because of them methods \emph{are not mere transitions}.
Instead, they give structure to the interactions between discrete computations and continuous behavior, as we have shown in Fig.~\ref{fig:bball2}.
Similar effects also occur with \abs{await} statements in a non-leading position, i.e., not the first statement in a method.
We assume that all methods are suspension-leading, i.e., each method starts with an \abs{await} statement with a \abs{diff} guard.
This is easily achieved by adding \abs{await diff true} if a method is not suspension-leading without changing the behavior.
We \emph{do} allow further \abs{await} statements within the method body.

\paragraph{Controllers.}
A useful pattern, already applied in Ex.~\ref{ex:bball}, are controllers. 
Controllers are methods which are always potentially schedulable, i.e., at every point in time, 
there is one process of such a method that can be scheduled if the leading guard expression holds.
\begin{definition}[Controller]
A method is a \emph{controller} if it (1) starts with an \abs{await} (2) contains no other suspensions and neither \abs{get} nor \abs{duration} statements, (3) ends with a recursive call, (4) is called only from the initial block.
\end{definition}
Because of (2), a controller runs instantaneous and because of (3) and (4) it is always potentially schedulable. Analyzing the leading guard of (1) allows one to precisely see under which conditions it can be scheduled.

In Ex.~\ref{ex:bball} class \abs{Tank} is controlled by \abs{up} and \abs{down}.
Controllers are a natural modeling pattern already for discrete systems: the largest Active Object modeling case study (modeling railway operations~\citep{scp} in \texttt{ABS}) uses controllers for almost all active classes.

\subsection{Runtime Semantics}\label{sec:lang:semantics}
The runtime semantics of \HABS is described by an evaluation function for expressions and guards, and a structural operational semantics for the transitions, i.e., rewrite rules on runtime syntax. 
Runtime syntax represents the state of a program. The \HABS semantics are fairly standard, except that (1) we use the timed extension of \cite{BjorkBJST13}, which allows
us to keep track of time in the semantics, and (2) that time advance changes not only the time, but also the object state, because it evolves according to the \abs{physical} block.
We first introduce runtime syntax and the evaluation function.

\paragraph{Runtime Syntax.}
Runtime syntax is the language for representation of the program state at runtime.
To fully describe the program state, runtime syntax uses an extended syntax for statements and uses a special \abs{suspend} statement, which deschedules a process.

\begin{figure}[tbh]
\begin{align*}
\mathit{tcn} &::= \mathsf{clock}(\xabs{e})~\mathit{cn}\qquad
\mathit{cn} ::= \mathit{cn}~\mathit{cn} \ssep \mathit{fut} \ssep \mathit{msg} \ssep \mathit{ob}
&&\text{(Timed) Configurations}
\\
\mathit{ob} &::= (o, \rho, {\mathit{ODE}, f}, \mathit{prc}, q)\qquad q ::= \mathit{prc} \cdot q \sep \varepsilon
&&\text{Objects, Process Queues}
\\
\mathit{msg} &::= \mathsf{msg}(o, \xabs{m}, \many{\xabs{e}}, \mathit{fid})\qquad
\mathit{fut} ::= \mathsf{fut}(\mathit{fid},\xabs{e})
&&\text{Messages, Futures}
\\
\mathit{prc} &::= (\tau,\mathit{fid},\xabs{rs})\ssep\bot\qquad
\xabs{rs} ::= \xabs{s} \ssep \xabs{suspend;s}
&&\text{Processes, Runtime Statements}
\end{align*}
\caption{Runtime syntax of \HABS.}
\label{fig:cfg}
\end{figure}

As values we use literal expressions and object identifiers.
\begin{definition}[Runtime Syntax~\citep{BjorkBJST13}]\label{runtimesyntax}
Let $\mathit{fid}$
range over future names, $o$ over object identities,
and $\rho,\tau,\sigma$ over stores (i.e., maps from fields or variables
to values). 
The runtime syntax of \HABS is given by the grammar in Fig.~\ref{fig:cfg}, where $\many{\cdot}$ denotes lists.
\end{definition}
The runtime syntax consists of a timed configuration that manages the global time, and a bag of messages (no yet handled method calls), futures and objects.

Formally, a timed configuration $\mathit{tcn}$ has a clock $\mathsf{clock}$ with
the current time, as an expression of \abs{Real} type and an object
configuration $\mathit{cn}$.  An object configuration 
consists of messages $\mathit{msg}$, futures $\mathit{fut}$ and
objects $\mathit{ob}$.  A message
$\mathsf{msg}(o, \xabs{m}, \many{\xabs{e}}, \mathit{fid})$ contains callee $o$, called method \abs{m}, passed
parameters $\many{\xabs{e}}$ and the generated future $\mathit{fid}$.  A future
$\mathsf{fut}(\mathit{fid},\xabs{e})$ connects the future name $\mathit{fid}$
with its return value \abs{e}.  An object
$(o, \rho, \mathit{ODE}, f, \mathit{prc}, q)$ has
an identifier $o$, an object store $\rho$, the current dynamic
$f$, an active process $\mathit{prc}$ and a set of inactive
processes $q$.  $\mathit{ODE}$ is taken from the class declaration.  A
process is either terminated $\bot$ or has the form
$(\tau,\mathit{fid},\xabs{rs})$: the process store $\tau$ with the current state of
the local variables, its future identifier $\mathit{fid}$, and the statement $\xabs{rs}$
left to execute. Composition $\mathit{cn}_1~\mathit{cn}_2$
is a commutative and associative concatenation.

\paragraph{Evaluation of Guards and Expressions.}
Guards and expressions are evaluated to values, given a store that maps all accessed fields and variables to values.
In our setting, 
expressions and guards are additionally evaluated with respect to a given time offset and continuous dynamics.
This is necessary for the transition system, which must compute how much time may advance before a differential guard is evaluated to true to determine the next
global time advance.
First, we give some technical definitions for stores and ODEs.

Given fixed initial values, the function $f$ is the solution of an ODE\footnote{For simplicity here we assume a unique solution, or, if the ODE has multiple solutions, a deterministic procedure to select one. The original \HABS semantics of \cite{lites} is defined more generally.} and $f(0) = \sigma$ 
with $\mathbf{dom}(\sigma) = \mathbf{dom}(\rho) \cup \mathbf{dom}(\tau)$, $\forall x \in \mathbf{dom}(\rho).~\sigma(x) = \rho(x)$
and $\forall x \in \mathbf{dom}(\tau).~\sigma(x) = \tau(x)$.
We denote this composition of functions with $\sigma = \rho \oplus \tau$.
Non-physical fields do not evolve.
\begin{definition}[Expressions]
Let $f$ be a continuous dynamics of class \abs{C}, i.e., a mapping from $\mathbb{R}^+$ to stores and $\sigma = f(0)$ the current state.
Let $\xabs{f}_p$ be a physical field and $\xabs{f}_d$ a non-physical field.
The semantics of fields $\xabs{f}_p$, $\xabs{f}_d$, unary operators ${\sim}$ and
binary operators
$\oplus$
after $t$ time units is defined as follows:
\begin{align*}
\sem{\xabs{f}_d}_\sigma^{f,t} &= \sigma(\xabs{f}_d) \qquad\qquad
\sem{\xabs{f}_p}_\sigma^{f,t} = f(t)(\xabs{f}_p)\\
\sem{{\sim}\xabs{e}}_\sigma^{f,t} &= {\sim}\sem{\xabs{e}}_\sigma^{f,t}\qquad\sem{\xabs{e} \oplus \xabs{e'}}_\sigma^{f,t} = \sem{\xabs{e}}_\sigma^{f,t} \oplus \sem{\xabs{e'}}_\sigma^{f,t}
\end{align*}
\end{definition}

Evaluation of guards is defined in two steps: 
First, the \emph{maximal time elapse} is computed. I.e., the maximal time that may pass without the guard expression evaluating to true.
For differential guards \abs{diff e} this is the minimal time until \abs{e} becomes true.
Then, the guard evaluates to true if no time advance is needed.

\begin{definition}[Guards]\label{def:semant-diff-guards}
Let $f$ be a continuous dynamic of object $o$ in state $\sigma$. The maximal time elapse {\normalfont($\mathit{mte}$)} and evaluation of guards is given in Fig.~\ref{fig:guard}.
\begin{figure}
\begin{align*}
\sem{\xabs{g}}_\sigma^{f,0} = \text{\normalfont\text{true}} &\iff \mathit{mte}_\sigma^f(\xabs{g}) \leq 0\qquad\qquad\mathit{mte}_\sigma^f(\xabs{duration(e)}) = \sem{\xabs{e}}^{f,0}_{\sigma} \\ 
\mathit{mte}_{\sigma}^f(\xabs{e?}) &= \cased{0 &\text{ if $\sem{\xabs{e}}$ is resolved} \\ \infty & \text{ otherwise}}\qquad\qquad
\mathit{mte}_\sigma^f(\xabs{diff e}) = \argmin_{t\geq 0}\left(\sem{\xabs{e}}_\sigma^{f,t} = \text{\normalfont\text{true}}\right)
\end{align*}
\caption{Semantics of guards}
\label{fig:guard}
\end{figure}
\end{definition}

\begin{figure*}
\noindent\scalebox{0.85}{\begin{minipage}{\textwidth}
\begin{align*}
\rulename{1}\ &\big(o, \rho,{\mathit{ODE}, f}, (\tau,\mathit{fid},\xabs{await g;s}), q\big) 
~\rightarrow~ \big(o, \rho,{\mathit{ODE}, f}, (\tau,\mathit{fid},\xabs{suspend;await g;s}),q\big)\\
\rulename{2}\ &\big(o, \rho,{\mathit{ODE}, f}, (\tau,\mathit{fid},\xabs{suspend;s}), q\big) 
~\rightarrow~ \big(o, \rho,{\mathit{ODE}, \mathsf{sol}(\mathit{ODE},\rho)}, \bot, q\cdot(\tau,\mathit{fid},\xabs{s})\big)\\
\rulename{3}\ &\big(o, \rho,{\mathit{ODE}, f}, \bot, q\cdot(\tau,\mathit{fid},\xabs{await g;s})\big)
~\rightarrow~ \big(o, \rho,{\mathit{ODE}, f}, (\tau,\mathit{fid},\xabs{s}),q\big) \\
&\hspace{5mm}\text{if $\sem{g}_{\rho\oplus\tau}^{f,0}$ = \text{\normalfont\text{true}}}\\
\rulename{4}\ &\big(o, \rho,{\mathit{ODE}, f}, \bot, q\cdot(\tau,\mathit{fid},\xabs{s})\big)
~\rightarrow~ \big(o, \rho,{\mathit{ODE}, f}, (\tau,\mathit{fid},\xabs{s}),q\big) \\
&\hspace{5mm}\text{if $\xabs{s}$ does not start with an \abs{await}}\\
\rulename{5}\ &\big(o, \rho,{\mathit{ODE}, f}, (\tau,\mathit{fid},\xabs{return e;}), q\big) 
~\rightarrow~ \big(o, \rho,{\mathit{ODE}, \mathsf{sol}(\mathit{ODE},\rho)}, \bot, q\big)~\mathsf{fut}\big(\mathit{fid},\sem{\xabs{e}}^{f,0}_{\rho\oplus\tau}\big)\\
\rulename{6}\ &\big(o, \rho,{\mathit{ODE}, f}, (\tau,\mathit{fid},\xabs{v = e.get;s}), q\big)~\mathsf{fut}\big(\mathit{fid},\xabs{e'}\big) 
~\rightarrow~ \big(o, \rho,{\mathit{ODE}, f}, (\tau,\mathit{fid},\xabs{v = e';s}), q\big) \\
&\hspace{5mm}\text{if }\sem{\xabs{e}}^{f,0}_{\rho\oplus\tau} = \mathit{fid}\\
\rulename{7}\ &\big(o, \rho,{\mathit{ODE}, f}, (\tau,\mathit{fid},\xabs{v = e!m(e}_1,\dots\xabs{e}_n\xabs{);s}), q\big) \\
&~\rightarrow~ \big(o,\rho,{\mathit{ODE},f},(\tau[\xabs{v}\mapsto \mathit{fid}_2],\mathit{fid},\xabs{s}),q\big) ~\mathsf{msg}\big(\sem{\xabs{e}}^{f,0}_{\rho\oplus\tau},\xabs{m},(\sem{\xabs{e}_1}^{f,0}_{\rho\oplus\tau},\dots,\sem{\xabs{e}_n}^{f,0}_{\rho\oplus\tau}),\mathit{fid}_2\big)\\
&\hspace{5mm}\text{where }\mathit{fid}_2\text{ is fresh}\\
\rulename{8}\ &\big(o, \rho,{\mathit{ODE}, f}, \mathit{prc}, q\big)~\mathsf{msg}(o, \xabs{m}, \many{\xabs{e}}, \mathit{fid}) 
~\rightarrow~ \big(o,\rho,{\mathit{ODE},f},\mathit{prc},q\cdot(\tau_{\xabs{m},\many{\xabs{e}}},\mathit{fid},\statement_\xabs{m})\big)\\
\rulename{9}\ &\big(o, \rho,{\mathit{ODE}, f}, (\tau,\mathit{fid},\xabs{v = new C(}\many{e}\xabs{);s}), q\big) \\
&~\rightarrow~ \big(o,\rho,{\mathit{ODE},f},(\tau[\xabs{v}\mapsto o'],\mathit{fid},\xabs{s}),q\big) 
~\big(o',\rho_{\classname,\many{e}},{\mathit{ODE}_\classname},\bot,(\tau',\mathit{fid}',\statement_\xabs{C}),\epsilon\big) \\
&\hspace{5mm}\text{where $o', \mathit{fid}'$ are fresh}
\end{align*}
\end{minipage}}
\caption{Interaction rules for \HABS objects.}
\label{fig:ntsem}
\end{figure*}

With the evaluation function, we can now give the transition system for \HABS. We split the rules into two groups:
the \emph{discrete transition system} describes the steps taken by a single object at a fixed point in time, while
the \emph{continuous transition system} describes the steps of the overall system and possibly advances the time.

\paragraph{Discrete Transition System}
The discrete transition system has two kinds of rules: \emph{interaction rules} that execute statements that 
interact with the environment and implement the concurrency model, and \emph{internal rules} that perform local steps in the active process of the object.

Fig.~\ref{fig:ntsem} gives the interaction rules for the semantics
of a single object. The first three rules handle the \abs{await} statement:
The rule \rulename{1} introduces a \abs{suspend} statement in front of an
\abs{await} statement.  Rule \rulename{2} consumes a \abs{suspend}
statement and moves a process into the queue of its object---at this
point, the ODEs are solved translated into some dynamics with $\mathsf{sol}$, with the current state $\rho$ as the initial conditions.
Rule \rulename{3} activates a process
with a following \abs{await} statement, if its guard evaluates to
true. 
Rule \rulename{4}
activates
a process with any other non-\abs{await} statement.  
Rule \rulename{5} realizes a termination: the active process is removed and the future is resolved by being added to the configuration, including the return value of the ODEs, and 
Rule \rulename{6} realizes a future read: if the read future is resolved, i.e., a \texttt{fut} component is in the configuration that contains the correct future identifier, then
the contained value is copied into the active process. 
Rule \rulename{7} is a method call that generates a message and \rulename{8} starts a process from a message, where (1) $\tau_{\xabs{m},\many{\xabs{e}}}$ is the initial local store
that maps each parameter of of \abs{m} to the corresponding concrete, passed parameter in declaration order and each local variable to the default value of its type, and (2) $\statement_\xabs{m}$ is the method body of \abs{m}.
Rule \rulename{9} creates a new object. The new heap $\rho_{\classname,\many{e}}$ assigns each field its initial value according to the initializing expression in the field declarations and
the object is created with the statement of the constructor ($\statement_\classname$) and the ODEs of the class ($\mathit{ODE}_\classname$).
Rules \rulename{6}, \rulename{7} and \rulename{9} have analogous equivalences for assignments for fields, where $\rho$ is updated instead of $\tau$.
The system is non-deterministic: rules \rulename{3} and \rulename{4} are free to choose \emph{any} schedulable process and \rulename{8} can process messages in any order.
Furthermore, the focus object for the rules is chosen non-deterministically.

The internal rules are given in Fig.~\ref{fig:ntsem:intern}. This part of \HABS is a standard \texttt{while}-language with an additional store, we only note that we ignore scopes of local variables for brevity's sake and initialize them upon process creation.
\begin{figure*}
\noindent\scalebox{0.85}{\begin{minipage}{\textwidth}
\begin{align*}
\rulename{10}\ &\big(o, \rho,{\mathit{ODE}, f}, (\tau,\mathit{fid},\xabs{v = e;s}), q\big) \\
~&\rightarrow~ \big(o, \rho,{\mathit{ODE}, f}, (\tau[\xabs{v} \mapsto \sem{\xabs{e}}^{f,0}_{\rho\oplus\tau}],\mathit{fid},\xabs{s}),q\big)
\hspace{40mm}\text{if \xabs{e} contains no call or \abs{get}}\\
\rulename{11}\ &\big(o, \rho,{\mathit{ODE}, f}, (\tau,\mathit{fid},\xabs{this.g = e;s}), q\big) \\
~&\rightarrow~ \big(o, \rho[\xabs{g} \mapsto \sem{\xabs{e}}^{f,0}_{\rho\oplus\tau}],{\mathit{ODE}, f}, (\tau,\mathit{fid},\xabs{s}),q\big)
\hspace{40mm}\text{if \xabs{e} contains no call or \abs{get}}\\
\rulename{12}\ &\big(o, \rho,{\mathit{ODE}, f}, (\tau,\mathit{fid},\xabs{while(e)}\{\xabs{s}\}\xabs{s'}), q\big) 
~\rightarrow~ \big(o, \rho,{\mathit{ODE}, f}, (\tau,\mathit{fid},\xabs{if(e)}\{\xabs{s;while(e)}\{\xabs{s}\}\}\xabs{s'}),q\big)\\
\rulename{13}\ &\big(o, \rho,{\mathit{ODE}, f}, (\tau,\mathit{fid},\xabs{if(e)}\{\xabs{s}\}\xabs{else}\{\xabs{s'}\}\xabs{s''}), q\big)
~\rightarrow~ \big(o, \rho,{\mathit{ODE}, f}, (\tau,\mathit{fid},\xabs{s;s''}),q\big)\\
&\hspace{5mm}\text{if $\sem{\xabs{e}}^{f,0}_{\rho\oplus\tau} = \mathsf{true}$}\\
\rulename{14}\ &\big(o, \rho,{\mathit{ODE}, f}, (\tau,\mathit{fid},\xabs{if(e)}\{\xabs{s}\}\xabs{else}\{\xabs{s'}\}\xabs{s''}), q\big) 
~\rightarrow~ \big(o, \rho,{\mathit{ODE}, f}, (\tau,\mathit{fid},\xabs{s;s'}),q\big)\\
&\hspace{5mm}\text{if $\sem{\xabs{e}}^{f,0}_{\rho\oplus\tau} = \mathsf{false}$}
\end{align*}
\end{minipage}}
\caption{Internal rules for \HABS objects.}
\label{fig:ntsem:intern}
\end{figure*}

\paragraph{Continuous Transition System.}
The continuous transition system contains two rules
shown in Fig.~\ref{fig:tsem}.
Rule \rulename{i} realizes a step of an object without advancing
time by using the discrete transition system for some object. Only if rule \rulename{i} is not applicable, rule \rulename{ii} can
be applied. It computes the global maximal time elapse and advances the time in the clock and all objects.  

\begin{figure}[tbh]
\scalebox{0.85}{\begin{minipage}{\columnwidth}
\begin{align*}
\rulename{i}~&\mathsf{clock}(t)~\mathit{cn}~\mathit{cn}' \rightarrow \mathsf{clock}(t)~\mathit{cn}''~\mathit{cn}' &&\text{ with } \mathit{cn} \rightarrow \mathit{cn}''\\
\rulename{ii}~&\mathsf{clock}(t)~\mathit{cn} \rightarrow \mathsf{clock}(t+t')~\mathit{adv}(\mathit{cn},t') &&\text{ if \rulename{i} is not applicable}\text{ and }\mathit{mte}(\mathit{cn}) = t'\neq \infty
\end{align*}
\end{minipage}}
\caption{Timed semantics of \HABS configurations.}
\label{fig:tsem}
\end{figure}
\begin{figure}[tbh]
\scalebox{0.9}{\begin{minipage}{\columnwidth}
\begin{align*}
\mathit{mte}(\mathit{cn}~\mathit{cn}') &= \mathbf{min}(\mathit{mte}(\mathit{cn}),\mathit{mte}(\mathit{cn}')) 
\qquad\qquad \mathit{mte}(\mathit{msg}) = \mathit{mte}(\mathit{fut}) = \infty\\
\mathit{mte}(o, \rho, \mathit{ODE}, f, \mathit{prc}, q) &=\sem{\mathbf{min}_{q' \in q}(\mathit{mte}(q'),\infty)}_\rho
\qquad\qquad \mathit{mte}(\tau,\mathit{fid},\xabs{await g;s}) = \sem{\mathit{mte}(g)}_\tau \\
\mathit{mte}(\tau,\mathit{fid},\xabs{s}) &= \infty\text{ if }\xabs{s}\neq \xabs{await g;s'}\qquad\qquad\qquad
\mathit{mte}(\xabs{duration(e)}) = \xabs{e}\qquad\\
\mathit{mte}_\sigma^f(\xabs{diff e}) &= \argmin_{t\geq 0}~ \left(\sem{\xabs{e}}_\sigma^{f,t} = \text{\normalfont\text{true}}\right)\qquad\qquad
\mathit{mte}(\xabs{e?}) = \infty\\
\end{align*}
\end{minipage}}
\begin{align*}
\mathit{adv}_\mathit{prc}\big((\tau,\mathit{fid},\xabs{s}),f,t\big) &= (\tau,\mathit{fid},\xabs{s}) \qquad\text{ if }\xabs{s} \neq (\xabs{await duration(e); s'})\\
\mathit{adv}_\mathit{prc}\big((\tau,\mathit{fid},&\,\xabs{await duration(e);s}),f,t\big) = (\tau,\mathit{fid},\xabs{await duration(e+}t\xabs{);s})\\
\mathit{adv}_\mathit{heap}(\rho,f,t)(\xabs{f}) &=
\cased{
\rho(\xabs{f}) & \text{ if \abs{f} is not physical} \\
f(t)(\xabs{f})  & \text{ otherwise }}
\end{align*}
\caption{Auxiliary functions.}
\label{fig:aux}
\end{figure}

Fig.~\ref{fig:aux} shows the auxiliary functions and includes the full
definition of \textit{mte}.  
Note that \textit{mte} is not applied to
the currently active process, because, when \rulename{i} is not
applicable, it is currently blocking and, thus, cannot advance time.
State change during a time advance is handled by a family of functions
    $\mathit{adv}$ which are applied to all stores and objects. 
    We only give two members of the family: $\mathit{adv}_{\mathit{heap}}$ takes as parameter a store $\sigma$, the dynamics $f$ and a duration $t$,
    $\mathit{adv}_{\mathit{prc}}$ takes a process $\mathit{prc}$, the dynamics $f$ and a duration $t$.
    Both advance its first parameter by $t$ time units according to $f$.  
    The state evolves according to the current dynamics and the guards of \abs{duration} guards and statements are decreased by $t$ (if the \abs{duration} clause is part of the first statement of an unscheduled process).
    For non-hybrid Active
    Objects $\mathit{adv}_{\mathit{heap}}(\rho,t) = \rho$. 

The semantics trigger the guards as soon as possible. 
This is crucial for the simulation of \HABS. 
The more relaxed semantics of, e.g., hybrid automata, introduce additional non-determinism that is beneficial for verification but difficult to handle meaningfully for simulation.
We remind that only weak inequalities are allowed, i.e., an event boundary can always be constructed and the minimum indeed exists if the boundary is every hit.

\subsection{Runs and Traces}\label{sec:lang:traces}
So far we have described the operational semantics, which translate one runtime configuration into another.
Such sequences are called \emph{runs}.
We are, however, interested in \emph{traces}: mappings from the continuous time to the state of each object at that point in time.
To compute traces, one must fill in the gaps between the points in time when the operational semantics made a step.
We must not interpolate, as we have the used dynamics as part of the object state.

Formally, a trace $\theta$ is a mapping from $\mathbb{R}^+$ to states, meaning that at time $t$ the state of the program is $\theta(t)$.
We say that
$\mathsf{clock}(t_i)~\mathit{cn}_i$ is the final configuration at
$t_i$ in a run, if any other timed configuration
$\mathsf{clock}(t_i)~\mathit{cn}_i'$ is before it. This ensures that if the discrete execution generates several states at the same point in time, we consider the last one when it comes to safety.
\begin{definition}[Traces]\label{def:trace}
Let \textsf{Prgm} be a program. Its initial state configuration is denoted $\mathsf{clock}(0)~\mathit{cn}_0$~\citep{BjorkBJST13}.
A run of \textsf{Prgm} is a
(possibly infinite) reduction sequence
\[\mathsf{clock}(0)~\mathit{cn}_0 \rightarrow \mathsf{clock}(t_1)~\mathit{cn}_1 \rightarrow \cdots\]
A run is \emph{time-convergent} if it is infinite and $\lim_{i \rightarrow \infty} t_i < \infty$.
A run is \emph{locally terminated} if every process within terminates.

For each object $o$ occurring in the run, its \emph{trace} $\theta_o$ is defined as 

\scalebox{0.9}{\begin{minipage}{\columnwidth}
\begin{align*}
\theta_o(x) = \cased{
\mathit{undefined} &\text{if $o$ is not created yet}\\
\rho &\text{if $\mathsf{clock}(x)~\mathit{cn}$ is the final configuration at $x$}\\ 
&\text{ and $\rho$ is the state of $o$ in $\mathit{cn}$.}\\
\mathit{adv}_{\mathit{heap}}(\rho,f,x-y) &\text{if there is no configuration at $\mathsf{clock}(x)$}\\
&\text{and the last configuration was at $\mathsf{clock}(y)$}\\
&\text{with state $\rho$ and dynamic $f$ }
}
\end{align*}
\end{minipage}}
\end{definition}
We normalize all traces and let them start with 0 by shifting all states by the time the object is created.

\begin{example}
Consider Fig.~\ref{fig:steps} which illustrates the semantics of tank from Ex.~\ref{ex:bball}, starting at time 1.
The store is $\rho_1 = \{\xabs{level} \mapsto 4,\xabs{drain} \mapsto -1, \xabs{log} \mapsto \mathit{l}\}$ and two processes are suspended.
The one for \xabs{up} is denoted $\mathit{prc}$, the one for \xabs{down} has the remaining statement $\statement^+_{\xabs{down}}$, which is the whole method body.

Nothing can execute without advancing time, so time is advanced by 1 time unit (using rule (ii)) until the store changes to 
$\rho_1 = \{\xabs{level} \mapsto 3,\xabs{drain} \mapsto -1,\xabs{log} \mapsto \mathit{l}\}$.

This enables rule (3) to schedule the process for \xabs{down}, where the \abs{await} statement is removed: $\statement_{\xabs{down}}$ is the method body without the leading suspension.
Finally, rule (7) is used to generate a message to call the other object. $\statement^-_{\xabs{down}}$ is the method body without the leading suspension and without the call statement to \abs{log}.
\begin{figure}
\begin{align*}
&\mathsf{clock}(1)~\Big(o,\rho_1,\mathit{ODE},f,\bot,\big\{(\emptyset,\mathit{fid}_1,\statement^+_{\xabs{down}}),\mathit{prc}\big\}\Big)\\
\xrightarrow{(ii)}~&\mathsf{clock}(2)~\Big(o,\rho_2,\mathit{ODE},f,\bot,\big\{(\emptyset,\mathit{fid}_1,\statement^+_{\xabs{down}}),\mathit{prc}\big\}\Big)\\
\xrightarrow{(3)}~&\mathsf{clock}(2)~\Big(o,\rho_2,\mathit{ODE},f,(\emptyset,\mathit{fid}_1,\statement_{\xabs{down}}),\big\{\mathit{prc}\big\}\Big)\\
\xrightarrow{(7)}~&\mathsf{clock}(2)~\Big(o,\rho_2,\mathit{ODE},f,(\emptyset,\mathit{fid}_1,\statement),\big\{\mathit{prc}\big\}\Big)\quad\mathsf{msg}(o,\xabs{log},\epsilon,\mathit{fid}_2)
\end{align*}
\caption{Example semantics of a tank from Ex.~\ref{ex:bball}.}
\label{fig:steps}
\end{figure}
\end{example}
We identify \HABS variables and fields with \ddl variables and, thus, consider combined states $\rho \oplus \tau$ as valuations.

\section{Post-Regions for Distributed Systems}\label{sec:region}
Having fixed the concurrency model, we now present post-region based verification for it. 
Several notions from post-regions for the concurrent programs in Sec.~\ref{sec:concurrent} carry over.

First, let us make more precise what we mean by object invariants. 
Classically, an object invariant has to hold whenever a method starts or ends. 
As discussed, this is not sufficient for hybrid systems, because whenever a method terminates, the dynamics may change the state during a time advance (i.e., when rule \rulename{ii} in Fig.~\ref{fig:tsem} is applied) and result in a state where the invariant does not hold. We, thus, introduce object invariants as a property that has to hold (1) whenever a method starts or ends and (2) whenever time advances. This also includes the case
where time advances but a process is active. This ensures that an object is safe, even if its discrete part is deadlocked.

Having fixed the notion of object invariants, we can now examine how this affects traces. 
For simplicity, let us ignore \abs{get} and \abs{duration} statements for now.
It is clear that we have to verify the state directly after a method releases control. 
Additionally, we must verify states that are part of the trace before the next process is scheduled.
We must, thus, express that the invariant is preserved by the dynamics until the next process is scheduled.

To do so, it is critical to give an overapproximation of the states following a suspension/termination until the next method/block starts. 
If more about these states is known, then the verification becomes \emph{more precise} but \emph{less robust} against changes.
We give three techniques, each with a different trade-off between precision and robustness.

\begin{description}
\item[Basic Regions]
As a baseline, we can use $\mathsf{true}$ to describe all states. We introduce this in Sec.~\ref{ssec:triv} as \emph{basic} regions.
This technique is the least precise, but the most robust one.

\item[Locally Controlled Regions]
A more precise approximation is possible by analyzing the calls made in a method.
We compute for each suspension and termination point the set of methods which are guaranteed to be called \emph{on the same object} at this point.
Using this set, we can approximate when the next process will run on this object: when the leading guard expression of the called method holds.

It suffices to show that the object remains safe until some of these guards holds. The object may schedule a process earlier (e.g., if called from the outside), but this process may then also assume the invariant. 
We introduce this as \emph{locally controlled} regions in Sec.~\ref{ssec:uniform}.
This technique is more precise than using basic regions, but less robust: If a method \methodname is removed, all methods calling \methodname must be reproven because their post-regions changes, even if method \methodname has no contract.

\item[Structurally Controlled Regions]
Finally, if a controller is used in the class, then its suspension conditions also describe conditions when the next (controller) process is scheduled at the latest.
We introduce this as \emph{structurally controlled} regions in Sec.~\ref{ssec:split}.
This technique is even more precise than using locally controlled regions, but is even less robust: removal of a controller requires to reprove all other methods, as 
all post-regions change.
\end{description}

\begin{figure}[bth]
\noindent\begin{abscode}
class Ball(){
  physical{  
    Real pos = 5: pos' = v;  
    Real v = 0: v' = a; //downwards velocity
    Real a = 9.81: a' = 0;
  }
  { this!ctrl(); }         
  Unit ctrl(){
    await diff pos <= 0; //Zeno behavior ignored for the illustration
    v = -v*0.9;
    this.ctrl();
  }
  Unit hold(){ v = 0; this!catch(); }
  Unit catch(){
      await diff v >= 3;
      v = 0;
  }
}
\end{abscode}
\begin{center}
\includegraphics[scale=0.425]{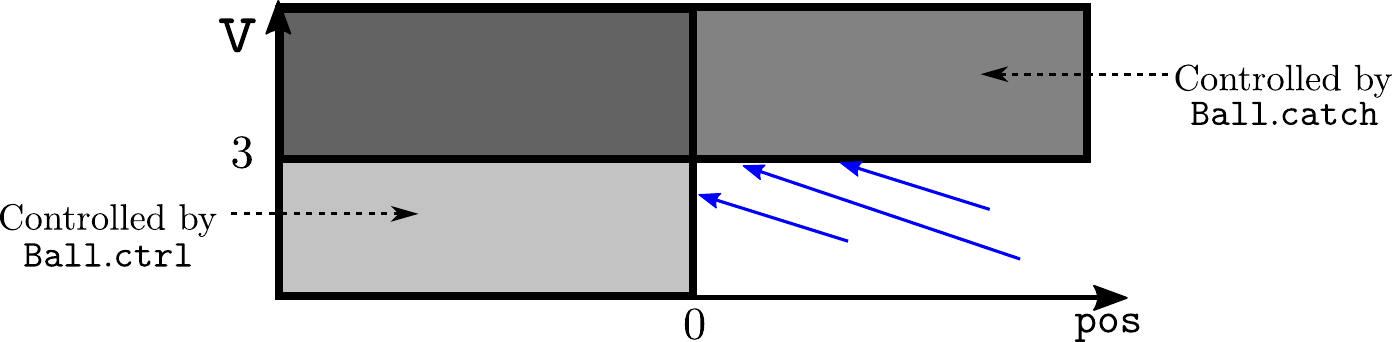}
\end{center}
\caption{Post-region of \abs{Ball.hold}. The arrows are (partial) traces after \abs{Ball.hold} terminates. 
The upper region is locally controlled by the call to \abs{Ball.catch}. The right region is structurally controlled by \abs{Ball.ctrl}.
Both these areas, and their overlap, are not relevant for verification: the post-region is only the white region.}
\label{fig:regions}
\end{figure}

\begin{example}
Consider the class and illustration in Fig.~\ref{fig:regions}. 
After \abs{Ball.hold} terminates, safety has to be established for all states until the next process is scheduled.
The upper shaded region cannot be part of these states --- as \abs{Ball.hold} called \abs{Ball.catch}, at the latest when \mbox{\abs{v >= 3}} holds this process will run.
This is the locally controlled region.
Similarly, when \mbox{\abs{pos <= 0}} holds, the controller will run. This is the right shaded region, the structurally controlled region. Thus, any of the dynamics (pictured as lines) must only establish safety for
the unshaded region. This is the post-region of \mbox{\abs{Ball.hold}}. 
We stress again that we do not verify the post-region itself, but the dynamics inside the post-region, from possible post-states of the method.
\end{example}
\paragraph{Specification.}
We specify objects with object invariants and methods with simple method contracts directly in \HABS using expressions.
Extracting post-conditions from futures is described by \cite{doa} and not specific to hybrid active objects.

Object invariants are specified as an annotation at the physical block.
Additionally, we specify a \emph{creation condition} that has to hold once an instance is created.
%
%
Both creation condition and object invariant can only refer to the fields of the specified class, not to the fields of referenced objects. We refer to the creation condition of a class \classname with $\mathsf{init}_{\classname}$ and to the invariant with $\mathsf{inv}_{\classname}$.

A method may be annotated with a pre-condition and a post-condition.
The pre-condition is denoted $\mathsf{pre}_\methodname$ and can only refer to the parameters of the method.
The post-condition is denoted with $\mathsf{post}_{\methodname}$ and may refer only to fields and the special symbol \abs{result} for the return value. 
Fig.~\ref{fig:trivialfig} gives an example.
%

All notions of post-regions share the general form of proof obligations, which is sketched in formula~\ref{eq:intromain}.
However, we require additional elements in the proof obligations: non-leading \abs{await} statements can violate the object invariant, but \ddl has no mechanism
to verify such failures before the end of symbolic execution\footnote{In contrast to other dynamics logics, such as PDL~\citep{dynamic} or BPL~\citep{bpl}.}.
For this reason, and similarly for method pre- and post-conditions, we require additional terms.
To describe the general form of proof obligations for HABS, we introduce the notion of proof obligation schemes and sound post-regions in Sec.~\ref{ssec:general}.

\subsection{General Post-Regions}\label{ssec:general}
In this section we formalize post-region invariants for HAOs and introduce the notion of soundness for post-regions, i.e., a formalization when a way to generate post-regions
leads to a sound proof principle.
The target of our verification are object invariants, safety properties on single objects that hold whenever a process starts or time advances. As every object is strongly encapsulated, an invariant can only specify the fields of a single object. We denote the specified invariant for a class \abs{C} with $I_\xabs{C}$. For initialization, a constraint on the initial values of the externally initialized fields may be specified. This \emph{creation condition} is denoted $\mathsf{pre}_\xabs{C}$ and used a precondition for the constructor.

We give a proof obligation scheme that generates \ddl-formulas for all methods in a class. To show safety of a class, all generated \ddl-formulas must be proven valid.
The scheme is parametric in both object invariant and notion of post-region, which we formalize using \emph{post-region generators}.
We give a slightly simplified treatment of method calls and contracts here and only differ between methods that are always called at a certain suspension point, and all other methods. 
Their specification is as for Active Objects~\citep{doa}: the precondition $\mathsf{pre}_\xabs{C.m}$ is a first-order formula over the method parameters and the post-condition $\mathsf{post}_\xabs{C.m}$ is a first-order formula over the fields of a method and a special variable \abs{result} for the return value.

Before the proof obligation scheme, we define some auxiliary structures, as \HABS does not directly connect to \ddl-statements.

First, we use three extra program variables:
\begin{itemize}
\item \abs{t} to keep track of time.
\item \abs{cll} (for \emph{call}) to keep track of contract violations during execution of a method.
This is necessary, because the post-condition is evaluated at the end of the method -- intermediate failures must be remembered until then. 
\item \abs{result} is used to store the return value of a method. It can be used only in the post-condition of a method. 
\end{itemize}
We assume that there are no fields or local variables with these names.
Secondly, we define the following abbreviations:
\begin{itemize}
    \item The following statement set \abs{cll} to 1, i.e., fails the analysis.
    \[\mathsf{fail} = \xabs{cll :=}~1\]
    \item The following statement sets all fields, including physical fields, \emph{but not local variables} to new values.
    This is used to approximate suspension, where another process can run, but only change fields.
    \[\mathsf{havoc} = \xabs{f}_1\xabs{:= *;}\dots\xabs{f}_n\xabs{:= *;}\quad\text{\it for all fields $\xabs{f}_i$ }\]
    \item The following statement sets all \emph{physical} fields to new values.
    This is used to approximate time advance while the object is blocked.
    \[\mathsf{havoc}^\mathsf{ph} = \xabs{f}_1\xabs{:= *;}\dots\xabs{f}_n\xabs{:= *;}\quad\text{\it for all physical fields $\xabs{f}_i$ }\]
\end{itemize}

Third, the following function $\mathsf{pr}$ generates a formula that a certain invariant $\mathsf{inv}$ holds for dynamics $\mathsf{ode}$ in post-region $\psi$:
\[
\mathsf{pr}(\psi, \mathsf{inv}, \mathsf{ode}) = \mathsf{inv} \wedge \big[\xabs{t} \xabs{:=} 0;\{\mathsf{ode},\xabs{t}' = 1\& \psi\}\big]\mathsf{inv}\\
\]
Whenever $\mathsf{inv}$ and $\mathsf{ode}$ are understood, we just write $\mathsf{pr}(\psi)$.
Finally, the proof obligation scheme is parametric in a \emph{post-region generator}, which is a map $\Psi$, that maps 
methods names and program point-identifiers to post-regions, 
i.e, formulas over the physical fields of a class and \abs{t}. 

We can now define our proof obligation scheme.
\begin{definition}[Proof Obligation Scheme]\label{def:scheme}
Let \abs{Prgm} be a program. For each class \abs{C} in the program, let $\mathsf{ode}_\xabs{C}$ be its dynamics and $\statement_\xabs{init}$ the code of the constructor.
For every method \abs{C.m} let $\statement_\xabs{C.m}$ be the method body.
Let $\statement_\xabs{main}$ be the statement of the main block.
The proof obligation scheme $\iota^\Psi$ for a post-region generator $\Psi$ is defined as follows.

For every class \abs{C}, there is one formula 
\[\iota_\xabs{C.init}^\Psi \equiv \mathsf{pre}_\xabs{C} \wedge \xabs{cll} \doteq 0\rightarrow \big[\mathsf{trans}(\statement_\xabs{init})\big]\big(\xabs{cll} \doteq 0 \wedge \mathsf{pr}(\Psi(\xabs{C.init}), I_\xabs{C}, \mathsf{ode}_\xabs{C})\big)\]
for each method \methodname in \xabs{C} one formula
\[\iota_\xabs{C.m}^\Psi \equiv I_\classname \wedge \mathsf{pre}_{\xabs{C.m}} \wedge \xabs{cll} \doteq 0 \rightarrow \big[\mathsf{trans}(\statement_\xabs{C.m})\big]\big(\xabs{cll} \doteq 0 \wedge\mathsf{post}_{\xabs{C.m}} \wedge  \mathsf{pr}(\Psi(\xabs{C.m}), I_\xabs{C}, \mathsf{ode}_\xabs{C})\big)\]
and for the main block the formula
\[ \iota_\xabs{main}^\Psi \equiv\mathsf{cll} \doteq 0 \rightarrow  \big[\mathsf{trans}(\statement_\xabs{main})\big]\xabs{cll} \doteq 0\]
The translation $\mathsf{trans}$ of \HABS statements into \ddl statements is given in Fig.~\ref{fig:trans}.
\begin{figure}[t!bh]
\scalebox{1}{\begin{minipage}{\columnwidth}
\begin{align*}
\mathsf{trans}(\statement_1;\statement_2) &= \mathsf{trans}(\statement_1)\xabs{;}\mathsf{trans}(\statement_2) \\
\mathsf{trans}(\xabs{if(e)}\{\statement_1\}\xabs{else}\{\statement_2\}) &= 
\mathsf{if}(\mathsf{trans}(\xabs{e}))\{\mathsf{trans}(\statement_1) \mathsf{else} \{\mathsf{trans}(\statement_2)\}\\
\mathsf{trans}(\xabs{while(e)}\{\statement\}) &= \mathsf{while}(\mathsf{trans}(\xabs{e}))\{\mathsf{trans}(\statement)\} \\
\mathsf{trans}(\xabs{v = e}) &= \xabs{v:=}~\mathsf{trans}(\xabs{e})\\
\mathsf{trans}(\xabs{return e}) &= \xabs{result :=}~\mathsf{trans}(e)
\end{align*}
\begin{align*}
\mathsf{trans}(\xabs{v = e.get}) &= 
\big\{\{?\mathsf{pr}(\mathsf{true})\}\xabs{;}\mathsf{havoc}^{\mathsf{ph}}\xabs{;}?\mathsf{inv}_\xabs{C}\xabs{;}\xabs{v :=*}\big\} \cup \big\{\{?\neg\mathsf{pr}(\mathsf{true}); \mathsf{fail}\}\xabs{;}\mathsf{havoc}^{\mathsf{ph}}\xabs{;}\xabs{v :=*}\big\}\\
\mathsf{trans}(\xabs{await}_\mathtt{p}~\xabs{g}) &= 
\phantom{\cup}\{?\mathsf{pr}(\Psi(\mathtt{p}) \wedge \tilde{\neg}\mathsf{trans}(\xabs{g}))\}\xabs{;}\mathsf{havoc}\xabs{;}?\mathsf{trans}(\xabs{g}) \wedge \mathsf{inv}_\xabs{C}\} \\
&\phantom{=}\cup 
\{?\neg\mathsf{pr}(\Psi(\mathtt{p}) \wedge \tilde{\neg}\mathsf{trans}(\xabs{g}))\xabs{;}\mathsf{fail}\}\xabs{;}\mathsf{havoc}\xabs{;}?\mathsf{trans}(\xabs{g})\} 
\end{align*}
\begin{align*}
\mathsf{trans}&(\xabs{v = e!m(e}_1,\dots\xabs{)}) =
\big\{\{?\mathsf{pre}_{\methodname}(\xabs{e}_1,\dots) \} \cup \{?\neg\mathsf{pre}_{\methodname}(\xabs{e}_1,\dots)
\xabs{;}\mathsf{fail} \}\big\}\xabs{;}
\xabs{v := *} \\
\mathsf{trans}&(\xabs{v = new C(e}_1,\dots\xabs{)}) =
\big\{\{?\mathsf{init}_{\classname}(\xabs{e}_1,\dots) \} \cup \{?\neg\mathsf{init}_{\classname}(\xabs{e}_1,\dots)
\xabs{;}\mathsf{fail} \}\big\}\xabs{;}
\xabs{v := *} \\
\mathsf{trans}&(\xabs{duration(e)}) = 
~\xabs{t := 0;}\big\{\{?\mathsf{pr}(\xabs{t} \leq \xabs{e})\} \cup \{?\neg\mathsf{pr}(\xabs{t} \leq\xabs{e}); \mathsf{fail}\}\big\}\xabs{;}\\
&\hspace{22.5mm}\xabs{t := 0;}\{\mathsf{ode},\xabs{t}' = 1 \& \xabs{t} \leq \mathsf{trans}(\xabs{e}) \}\xabs{;} ?\xabs{t} \geq \mathsf{trans}(\xabs{e}) 
\end{align*}
\begin{align*}
\mathsf{trans}(\xabs{f}) &= f\text{ , where $f$ is the \ddl variable representing field \xabs{f}}\\
\mathsf{trans}(\xabs{v}) &= v\text{ , where $v$ is the \ddl variable representing variable \xabs{v}}\\
\mathsf{trans}(\xabs{e}_1~\mathit{op}~\xabs{e}_2) &= \mathsf{trans}(\xabs{e}_1)~\mathit{op}~\mathsf{trans}(\xabs{e}_2)\\
\mathsf{trans}(\xabs{diff e}) &= \mathsf{trans}(\xabs{e})\\
\mathsf{trans}(\xabs{g}) &= \mathsf{true}\text{ if \abs{g} does not have the form \abs{diff e}}\\
\end{align*}
\end{minipage}}
        
\caption{Translation of \HABS-statements into \dL programs.}
\label{fig:trans}
\end{figure}

\end{definition}
If $\Psi(\cdot)$ is the same for all methods and suspension points with, and the class is understood, then we just write $\Psi$.

First, we examine the proof obligation for normal methods. 
The precondition $I_\classname \wedge \mathsf{pre}_{\xabs{C.m}} \wedge \xabs{cll} \doteq 0$ lets us assume the object invariant and the method precondition. The last term initializes the special variable \abs{cll} to $0$ -- i.e., no fault has occured.
The right-side expresses that after the execution of the method body the formula $\xabs{cll} \doteq 0 \wedge\mathsf{post}_{\xabs{C.m}} \wedge  \mathsf{pr}(\Psi(\xabs{C.m}), I_\xabs{C}, \mathsf{ode}_\xabs{C})$ holds. The first term expresses that no intermediate check failed the proof and \abs{cll} is still $0$.
The second term verifies the post-condition of the method and the last term ensures safety in the post-region.
As previously discussed, it takes the form $I_\xabs{C} \wedge \left[\xabs{t} \xabs{:=} 0;\{\mathsf{ode},\xabs{t}' = 1\& \psi\}\right] I_\xabs{C}$. It ensures that the invariant holds when the method terminated and that it is an invariant for the dynamics in a defined post-region. Variable \abs{t} is used to keep track of the time since the termination --
as we will see later, this is used to express post-regions that are time-bound, i.e., a process is guaranteed after at most $n$ time units.

The proof obligation for constructors is analogous but (1) does not assume the invariant, as it has not been established yet, and (2) assumes the creation condition as its precondition. A constructor also has no post-condition. The proof obligation for the main block is \emph{only} checking that all calls and object creations adhere to the respective precondition using \abs{cll}, as it does not runs in any object and has no additional specification.

The translation of statements into \ddl-programs works as follows.
We consider all fields as variables for the translation, and assume that each variable and field are program-wide unique.
The translation of sequence, branching, loops and assignment of side-effect free expressions to location is straightforward. 
The other statements are translated as follows:
\begin{itemize}
    \item Synchronization with \abs{get} first checks $\mathsf{pr}(\mathsf{true})$, if the formula does not hold, then verification fails. This models that during synchronization, time may pass and the invariant must, thus, hold. It is \emph{not} sound to use the post-region $\psi$ here: synchronization blocks the object, so no other process can run and take over. Furthermore, it may stay blocked for an unbound amount of time, so the invariant must hold for an unbound amount of time as well.
    Afterward, \abs{v} is set to a new, unknown value, as the return value in a future is not specified. Additionally, $\mathsf{havoc}^\mathsf{ph}$ is used to model that the physical fields may have changed during the synchronization. In case the check succeeds, the invariant is known to hold for the new values of the fields.
    \item Suspension with $\xabs{await}_i$ is similar, but uses both its post-region $\Psi(\mathsf{p})$ \emph{and} the guard, as the guard triggers reactivation.
    We use $\mathsf{havoc}?\xabs{;}\mathsf{trans}(\xabs{g})$ to set all fields (but not variables) to new values.
    For these new values only the guard is known to hold. The invariant is also known to hold, but only if the check succeeds.
    \item Method calls check the precondition of the called method in their check. Again, \abs{v} is set to a new, unknown value to model a fresh future.
    \item Object creation is analogous to method calls.
    \item Finally, blocking time advance is similar to synchronization using \abs{get}, with two differences: 
    First, while it is still not sound to use $\psi$, we may limit the time spend executing this statement by evaluating the guard.
    Second, instead of causing $\mathsf{havoc}$, we can precisely simulate the state change but advancing the dynamics for the amount of time given in the \abs{duration} statement.
\end{itemize}

There is an important detail concerning \abs{await duration} -- this guard is \emph{not} used to limit the time advance of the suspension, as it is \emph{not} guaranteed that the process after this time is indeed the one that runs next.
\begin{example}
Consider a method \methodname with the method body \abs{s = await duration(5);}. The post-region at the suspension is $t \leq 5$. 
It is, however, \emph{not} guaranteed that it will reactivate after 5 time units.
It may be reactivated later, e.g., if the object is blocked in the time span of $4$ to $6$ time units.
However, for this to happen, something must be scheduled in the post-region of \abs{s}, where the invariant is guaranteed.
Thus, the proof obligation expresses that (1) the \emph{next} process can assume the object invariant and (2)  upon reactivation the object invariant holds.
It does \emph{not} express that the process is rescheduled as soon as possible.
\end{example}

We now define soundness. Intuitively, the scheme is sound is validity of all contained formulas is sufficient to prove safety of the given invariants. We consider partial correctness~\citep{DBLP:journals/cacm/Hoare69}: we show safety if each process terminates and there are no deadlocks. 
\begin{definition}[Sound Proof Obligation Scheme]\label{def:soundscheme}
If the validity of all proof obligations from $\iota^\Psi$ implies that for all 
locally terminating, time-divergent runs, 
$\mathsf{inv}_{\xabs{C}}$ holds in every state of every trace of every object $o$ realizing any class \xabs{C}, whenever (1) $o$ is inactive or (2) time advances, and that  
the pre-condition of a method holds in every pre-state and the post-condition in every post-state,
then we say that $\iota^\Psi$ is sound.
\end{definition}
Our notion of safety differs from the one for discrete systems~\citep{DinO15,key} by adding condition (2): we require that the object stays safe whenever time advances, even if a method is already active. This is critical, as otherwise an object would be in an unsafe state, but would still be considered safe if it, for example, performs a non-suspending \abs{duration}.

Soundness of a proof obligation scheme depends on the post-regions: they must soundly overapproximate the states after a suspension and before the next process is scheduled. To make this precise, we introduce the notion of \emph{suspension-subtraces}.

A suspension subtrace decribes the trace between two methods/blocks whenever time advances.
\begin{definition}[Suspension-Subtraces]
Let \xabs{C.m} be a method in some program $\mathtt{Prgm}$.
Let $\trace_o$ be a trace, stemming from some run of $\mathtt{Prgm}$ for some object $o$ of \xabs{C} class. 
Let $i$ be the index in $\trace_o$ where some process of $\mathtt{m}$ terminates.
We say that $\trace_o^i$ is the suspension-subtrace of $\trace_o$, if it starts at $i$ and ends at (including) the time $j$ where the next non-trivial\footnote{I.e., a process that performs any action instead of descheduling immediately.} process is scheduled, with $i \neq j$. 
If no non-trivial process is scheduled afterwards, then $\trace_o^i$ is infinite.
Additionally, $\trace_o^i$ has a variable $\mathtt{t}$ with $\trace^i_o(0)(\mathtt{t}) = 0$ and $\mathtt{t}' = 1$. I.e., a clock that keeps track of the length of $\trace_i$.

The set of all suspension-subtraces of $\mathtt{m}$ in $\mathtt{Prgm}$ is denoted  $\Theta(\mathtt{m},\mathtt{Prgm})$.
The set of all suspension-subtraces of a program point-identifier $\mathtt{p}$ is $\Theta(\mathtt{p},\mathtt{Prgm})$ and defined analogously.
\end{definition}
Suspension-subtraces contain exactly the states where time advances and no process is active. 
The condition $i \neq j$ is used to ensure that time indeed advances for this trace and rules out traces containing a single state.
We remind that (non-trivial) post-regions are not used when time advances and a process is active, for example during the execution of a \abs{get} statement.
Thus, safety of a class can be reduced to analyzing suspension-subtraces, and we can give the following definition for sound post-regions.
\begin{definition}[Sound Post-Regions]\label{def:sound:region}
Let $\Psi$ be a post-region generator. 
Let $x_\xabs{C}$ range over the methods and program point identifiers within a class \abs{C}.
We say that $\Psi$ is a sound, if for every such $x$
(i.e., for every method \abs{m} within $C$, and every program point-identifier \texttt{p} within $C$)
every state of every suspension-subtrace is a model for $\Psi(x)$:
\[\forall \mathtt{Prgm}.~\forall \trace \in \Theta(x,\mathtt{Prgm}).~\forall i \leq |\trace|.~\trace[i] \models \Psi(x)\]
\end{definition}

Our main soundness theorem states that soundness of post-regions implies soundness of proof obligations schemes.
\begin{theorem}\label{thm:main}
If $\Psi$ is sound, then the proof obligation scheme of Def.~\ref{def:scheme} is sound in the sense of Def.~\ref{def:soundscheme}.
\end{theorem}
The proof, given in the appendix, follows the structure of the proof system for (non-hybrid) Active Objects based on Behavioral Program Logic~\citep{bpl}. It is shown that if every process can assume the invariant every time is scheduled, then it reestablishes the invariant upon descheduling. However, instead of relying on the invariant still holding since the descheduling of the last process, it must be shown that the dynamics have preserved it so far. This is exactly the condition of sound post-regions. Additionally, \abs{get} and \abs{duration} need to be considered because they may advance time. If the above property is established, the proof of the theorem is an induction on the length of trace, where the base case is that the main block and constructors do not depend on any invariant.

It is easy to see that $\mathsf{true}$, i.e., a basic post-region, is always sound, and we will investigate it further below.
The post-region $\mathsf{false}$ may also be sound, if the method in question always suspends or terminates in a situation where another process will be scheduled directly. For example,  consider the following two methods

\noindent\begin{abscode}
Unit m1(){ this!m2(); x = 5;}
Unit m2(){ await diff x >= 5; }
\end{abscode}

\noindent After a process of \abs{m1} terminates, time never advances before the called process of \abs{m2} (or another process) is scheduled. Thus, \abs{m1} has no suspension-subtraces.

Post-regions are compositional, it is possible to derive a post-region per situation and then join them.
For example, the two controllers in Ex.~\ref{ex:bball} define two different post-regions -- by the following proposition we can compute them independently. 
\begin{proposition}[Composition of Sound Post-Region Generators]\label{lem:compose} 
If $\Psi^1$ and $\Psi^2$ are sound, then the point-wise conjunction $\Psi^1 \wedge\Psi^2$ defined by 
\[(\Psi^1 \wedge\Psi^2)(x) = \Psi^1(x) \wedge\Psi^2(x)\]
is also sound.
\end{proposition}

\subsection{Basic Regions}\label{ssec:triv}
The basic region is the set of all states, described by the formula \abs{true}. Intuitively, this means that after a process terminates or suspends, the object remains 
safe without any further control necessary.
Basic regions are suited to verify objects which have to be stable without an explicit control loop or a specific order of method calls. 
They do not use any additional information about the program and can be a sound fallback if no better regions can be computed.
They are also a rather simple system to show the structure of the generated proof obligations in detail before we introduce more complex post-regions.
\begin{theorem}[Basic Post-Regions]\label{lem:trivial}
The post-region generator $\Psi^\mathsf{basic}$ that maps every method and program point-identifier to $\mathsf{true}$ is sound in the sense of Def.~\ref{def:sound:region}.
\end{theorem}
\begin{proof}
Every model is a model for $\mathsf{true}$.
\end{proof}

\begin{example}
Fig.~\ref{fig:trivialfig} describes a simple class where the field \abs{v} is specified to follow limited growth with limit \abs{bnd} and growth rate \abs{rate}.
It is specified that the value of \abs{v} is always between 0 and \abs{bnd}, and the rate is between 0 and 1.
Both rate and bound can be reset, but the bound can only be increased. The rate has to observe the specified interval.
The bound is set after a dynamic check, because the field is not visible from the outside, while the new growth rate is specified with a method pre-condition.
Safety relies on the fact that the method contract is adhered to.
\begin{figure}
\begin{abscode}
interface Element {	
 Unit inBound(Real nB);
 [HybridSpec: Requires(nR > 0 && nR < 1)] 
 Unit inRate(Real nR);
 [HybridSpec: Ensures(0 < result)]	
 Real outV();
}

[HybridSpec: 
 Requires(inV > 0 && inB > inV && inR < 1 && inR > 0)]
[HybridSpec: 
 ObjInv(v > 0 && bnd > v  && rate < 1 && rate > 0)]
class Element(Real inV, Real inR, Real inB) implements Element{
 physical{
  Real rate  = inR : rate'  = 0;
  Real bnd   = inB : bnd' = 0;
  Real v     = inV : v'     = rate*(bnd-v);
 }
 Unit inBound(Real nB){ if(nB >= bnd) bnd = nB; }
 Unit inRate(Real nR){ this.rate = nR; }
 Real outV(){ return this.v; }
}
\end{abscode}
\caption{Controlled limited growth.} 
\label{fig:trivialfig}
\end{figure}
The four proof obligations are as follows:
\begin{align}
\mathit{init}\wedge \xabs{cll} \doteq 0 &\rightarrow \big[?\mathsf{true};\xabs{bnd} := \xabs{inB};\xabs{rate} := \xabs{inR};\xabs{v} := \xabs{inV}\big]\mathit{phy}\\
I\wedge \xabs{cll} \doteq 0 &\rightarrow \big[?\mathsf{true};\mathtt{if}( \xabs{nB} \geq \xabs{bnd})\xabs{bnd} := \xabs{nB} \big]\mathit{phy}\\
I\wedge \xabs{cll} \doteq 0 &\rightarrow \big[?\mathsf{true};\xabs{result} := \xabs{v}\big]\big(\mathit{phy} \wedge 0<\xabs{result}\big)\\
I \wedge 0<\xabs{nR} <&1\wedge \xabs{cll} \doteq 0\rightarrow\big[?\mathsf{true};\xabs{rate} := \xabs{nR}\big]\mathit{phy}
\end{align}
Where
\begin{align*}
 I &\equiv \xabs{v} > 0 \wedge \xabs{bnd} >\xabs{v}  \wedge \xabs{rate} < 1 \wedge \xabs{rate} >0\\
 \mathit{init} &\equiv \xabs{inV} > 0 \wedge \xabs{inB} >\xabs{inV}  \wedge \xabs{inR} < 1 \wedge \xabs{inR} >0\\
 \mathit{phy} &\equiv\big((I\wedge \xabs{cll} \doteq 0\\
 &\hspace{4mm} \wedge[\xabs{rate'} = 0, \xabs{bnd'} = 0, \xabs{v'} = \xabs{rate}*(\xabs{bnd}-\xabs{v}) \& \mathsf{true}]I\big)
\end{align*}
\end{example}

The above proof obligations can be automatically closed by \texttt{KeYmaera X}.
We can improve the analysis slightly by excluding methods which obviously do not influence the invariant.
This is a simple version of frames~\citep{bubel}.
\begin{lemma}
Thm.~\ref{lem:trivial} holds also if no proof obligations are generated for methods \abs{C.m} that 
(a) do not assign to any field that is read in the invariant, any ODE, or is physical,
(b) make no method calls,
(c) create no object, and 
(d) have $\mathsf{post}_\xabs{C.m} = \mathsf{true}$
\end{lemma}
\begin{proof}
If such a method is run in a configuration with heap $\rho$, such that $\rho$ fulfills the invariant, it may only result in some heap $\rho_2$
which also fulfills the invariant, as no field in it has changed.
Similarly, the dynamics are unchanged.
Due to (d) the post-condition trivially holds and due to (b) and (c) the condition of \abs{cll} also holds.
\end{proof}

The notion of robustness for basic regions is the same as for discrete systems: only an added or modified method needs to be (re)proven.
\begin{lemma}\label{lem:modulartrivial}
Let $\xabs{C}$ be a safe class according to Thm.~\ref{lem:trivial}. 
Let $\xabs{C}^-$ be \xabs{C} with some method removed\footnote{I.e., it is removed from the class and all calls to it are replaced by \abs{skip}.} and $\xabs{C}^+$ be \xabs{C} with some  method added. 
\begin{enumerate}
\item
$\xabs{C}^-$ is safe. 
\item 
To verify $\xabs{C}^+$, only the proof obligation of the new method has to be shown.
\end{enumerate}
If another class calls the added method, only the proof obligation of the calling methods has to be shown to make the calling class safe.
\end{lemma}
\begin{proof}
The removal of a method does not influence the proof obligation for any other methods in its own class. 
Thus, the argument for safety of $\xabs{C}^-$ follows directly from the argument of proof of Thm.~\ref{lem:trivial} for \abs{C}. 
Similarly, if a class is added, the other methods still guarantee that any state reachable after their suspension or termination is safe, the only unsafe behavior that could occur must be introduced by the added method. 
The final point is straightforward: as the called method did not exist before, the calling method must be modified itself to target it and must, thus, be reproven.
\end{proof}

In comparison with the system we gave in for concurrent programs in Sec.~\ref{sec:concurrent}, there are some differences. 
The most obvious is that we have a more complex translation of statements, as \HABS is not directly based on \ddl programs. 
Other points are more subtle. First, we have \emph{two} layers of robustness:
one for methods (as for procedures in concurrent programs), and one for classes, which has no correspondence. Changes in the implementation of one class do not affect another,
only changes of method contracts do.
Second, in \HABS and unbounded number of processes can be executable (not active, but executable in the sense of suspended and potentially active), with an unbounded number of objects existing in the overall system. Here, modularity and post-regions shine: they do not only enable encapsulation of methods/procedures, but this encapsulation enables in turn an abstraction -- the concurrency model is completely outside \ddl proof obligations. In previous work, it was shown that encoding even a single class with strong restrictions on method (e.g., no suspension, at most one process per method is executable at any time), leads to the need of additional specification and a complex encoding of scheduling~\citep{lites}.

\subsection{Locally Controlled Regions}\label{ssec:uniform}

Basic post-regions use no information about the class or method for verification.
Our first addition to the system is to use the call structure for more precision, something that has no equivalent in the concurrent programs of Sec.~\ref{sec:concurrent}.
When verifying some method \methodname, methods called within \methodname are guaranteed to be in the process pool afterwards --- 
thus, we can check their guards to over-approximate when the next process will run. The negation of their guard expressions can be added to the post-region.
We call such regions \emph{locally controlled}, as each process locally starts another process to limit its post-region.

\begin{example}\label{ex:loco}
Consider the class in Fig.~\ref{fig:uniform}. It models a water tank with two event boundaries, one for the upper limit and one for the lower limit. 
The invariant holds, as whenever the water is rising, method \abs{up} is active and will react before the tank overflows and analogously for \abs{down}.

\begin{figure}
\begin{abscode}
class Tank {
    [HybridSpec: ObjInv("x >= 3 & x <= 10")]
    physical{ Real x = 5: x' = v; Real v = -1; v' = 0; }
    { this!down(); }
    Unit down(){ await x <= 3; v = 1; this!up(); }
    Unit up(){ await x >= 10; v = -1; this!down(); }
}
\end{abscode}
\caption{Water Tank with Local Control}
\label{fig:uniform}
\end{figure}
\end{example}

\noindent For formalizing locally controlled regions, we require external triggers.
\begin{definition}[External Trigger]
We use $\mathit{ttrig}_\xabs{g}$ to denote the external trigger of \abs{g}. 
If \abs{g} is a differential guard, then $\mathit{ttrig}_\xabs{g}$ is the translation of the guard expression.
If \abs{g} is a duration guard with parameter \abs{e} then
\[\mathit{ttrig}_\xabs{g} = \xabs{t} \geq \mathsf{trans}(\xabs{e})\]
Otherwise, $\mathit{ttrig}_\xabs{g} = \mathsf{false}$.
\end{definition}
We denote the external trigger of the leading guard of \methodname with $\mathit{ttrig}_{\methodname}$.

As another technical auxiliary, we fomalize the notion of a method that is guaranteed to be called up to a certain point.
\begin{definition}[Causality Graph]
Let statement \statement be a method body. 
A causality graph is a graph $(V,E, \mathsf{entry}, \mathsf{exit})$, with $\mathsf{entry}, \mathsf{exit} \in V, E = V \times V$ and
$\forall v \in V.~(v,\mathsf{entry}) \not\in E \wedge (\mathsf{exit}, v) \not\in E$.

The causality graph $\mathsf{cg}(\statement)$ of statement \statement is defined in Fig.~\ref{fig:causal}.
We assume that several syntactically identical copies of a statement at different points can be distinguished by implicit program-point identifiers.
\begin{figure}
\resizebox{\textwidth}{!}{
\begin{minipage}{\textwidth}
\begin{align*}
\mathsf{cg}(\xabs{while(e)}\{\statement\}) &= (V,E,\mathsf{in}, \mathsf{out})\text{, with}\\
&V = V_1 \cup \{\mathsf{in}, \mathsf{out}\}\\
&E = E_1 \cup \{(\mathsf{in},v) \sep v \in \mathsf{entry}_1\} \cup \{(v,\mathsf{in}) \sep v \in \mathsf{exit}_1\} \cup (\mathsf{in},\mathsf{out})\\
&\text{where }\mathsf{cg}(\statement_1) = (V_1,E_1, \mathsf{entry}_1, \mathsf{exit}_1)\\
\mathsf{cg}(\xabs{if(e)}\{\xabs{s}_1\}\xabs{else}\{\xabs{s}_2\}) &=  (V,E,\mathsf{in}, \mathsf{out})\text{, with}\\
&V = V_1 \cup V_2 \cup \{\mathsf{in}, \mathsf{out}\}\\
&E = E_1 \cup E_2 \cup \{(\mathsf{in},v) \sep v \in \mathsf{entry}_1 \cup \mathsf{entry}_2\} \cup \{(v,\mathsf{out}) \sep v \in \mathsf{exit}_1 \cup \mathsf{exit}_2\}\\
& \text{where }\mathsf{cg}(\statement_1) = (V_1,E_1, \mathsf{entry}_1, \mathsf{exit}_1), \mathsf{cfg}(\statement_2) = (V_2,E_2, \mathsf{entry}_2, \mathsf{exit}_2)\\
\mathsf{cg}(\statement_1\xabs{;}\statement_2) &= (V_1\cup V_2,E_1 \cup E_2 \cup (\mathsf{exit}_1,\mathsf{entry}_2),\mathsf{entry}_1,\mathsf{exit}_2) \\
& \text{where }\mathsf{cg}(\statement_1) = (V_1,E_1, \mathsf{entry}_1, \mathsf{exit}_1), \mathsf{cfg}(\statement_2) = (V_2,E_2, \mathsf{entry}_2, \mathsf{exit}_2)\\
\mathsf{cg}(\statement) &= (\{\statement\},\emptyset,\statement,\statement) \\
&\text{for any other \statement}
\end{align*}
\end{minipage}
}
\caption{Causality graph of a statement.}
\label{fig:causal}
\end{figure}

A method path of node $n$ is a path from either $\mathsf{entry}$ or \abs{await} node to $n$, that does not pass through any \abs{await} node.

A method \abs{m'} is \emph{guaranteed to be called at termination of \methodname}, if every method path of the $\mathsf{exit}$ node in $\statement_\xabs{m}$
contains a call to \abs{m'} on \abs{this}.
The set of all such methods is denoted $\mathsf{gcall}(\xabs{C.m})$.
A method \abs{m'} is \emph{guaranteed to be called} at $\xabs{await}_i$, if every method path of the $\xabs{await}_i$ node in $\statement_\xabs{m}$
contains a call to \abs{m'} on \abs{this}.
The set of all such methods is denoted $\mathsf{gcall}(\xabs{i})$.
\end{definition}

Finally, we can now give the theorem for locally controlled regions, which states that we can use the methods that are guaranteed to be called up to a suspension or termination point
to compute sound post-regions from their guards.
\begin{theorem}[Locally Controlled Regions]\label{lem:uniform}
Given a class \abs{C} and a method \abs{m}, let $\mathsf{calls}_{\xabs{C.m}}$ be the set of methods that are guaranteed to be called within \abs{C.m}.

The post-region generator $\Psi^\mathsf{local}$ that maps every method \abs{C.m} to 
\[
\bigwedge_{\methodname \in \mathsf{gcall}(\xabs{C.m})} \widetilde{\neg}\mathit{ttrig}_\methodname
\]
and every \abs{await} identifier $i$ to
\[
\bigwedge_{\methodname \in \mathsf{gcall}(\xabs{i})} \widetilde{\neg}\mathit{ttrig}_\methodname
\]
is sound.
\end{theorem}

\noindent The form of formula~\ref{eq:intromain} is now crucial. Contrary to basic regions, we cannot use a single modality by:
\[ [\statement](\mathsf{inv}_\classname \wedge [\mathsf{ode}_\classname\&\mathsf{true}]\mathsf{inv}_\classname)  \iff [\statement; \mathsf{ode}_\classname\&\mathsf{true}]\mathsf{inv}_\classname\]
because the post-region may be empty. 
In this case, 
e.g., if a method with a leading guard \abs{true} is called,
the second modality simplifies to $\mathsf{true}$. This would mean that without the additional check on $\mathsf{inv}_\classname$ in the post-condition of the first modality,
the called method would not be guaranteed the object invariant after all.
The used obligation still ensures the invariant in the post-state. 

Returning to our analysis of robustness, we observe that locally controlled post-regions are less compositional:
Removal of a method requires reproving all calling methods.
\begin{lemma}\label{lem:modulareuniform}
Let $\xabs{C}$ be a safe class according to Thm.~\ref{lem:uniform}. 
Let $\xabs{C}^-$ be \xabs{C} with some method $\methodname^-$removed and $\xabs{C}^+$ be \xabs{C} with some method added. 
\begin{enumerate}
\item
To verify $\xabs{C}^-$, only the proof obligations of methods calling $\methodname^-$ must be shown.
\item 
To verify $\xabs{C}^+$, only the proof obligation of the new method has to be shown.
\end{enumerate}
If a method (potentially from another class) calls $\methodname^+$, only the calling methods have to be reproven.
\end{lemma}
\begin{proof}
The only difference to Lem.~\ref{lem:modulartrivial} is the first point, which follows from the simple observation that if a method calls $\methodname^-$, 
then $\methodname^-$ can occur in some $\mathsf{CM}$ set and thus in its post-region. Without $\methodname^-$, the boundaries of the post-region are shifted, as
there is one message less to be processed and the calling method must, thus, be reproven for the new post-region.\qedhere
\end{proof}

One of the patterns that 
locally controlled regions can verify are classes with a single controller and no further methods.
\begin{example}\label{ex:ctank}
The controller in Fig.~\ref{fig:ctank} checks the water level of a tank every $\frac{1}{2}$ seconds,
so every process has to establish safety only for that time frame.
This example can be verified using locally controlled post-regions.
\begin{figure}[b!th]
\begin{abscode}
[HybridSpec: Requires("3.5<=inVal<=9.5")]
class TankTick(Real inVal){
  [HybridSpec: ObjInv("3<=x<=10 & -1<=v<=1")]
  physical{
    Real x = 5: x' = v;
    Real v = -1; v' = 0;
  }
  { this!ctrl(); }
  Unit ctrl(){
    await duration(1/2);
    if(x <= 3.5) v = 1;
    if(x >= 9.5) v = -1;
    this!ctrl();
  }
}
\end{abscode}
\caption{Timed water tank.}
\label{fig:ctank}
\end{figure}
\end{example}

They are also suitable for mutually recursive call structures -- Ex.~\ref{ex:loco} can be verified now.

\subsection{Structurally Controlled Regions}\label{ssec:split}

Locally controlled regions are not sufficient for multiple controllers, or one controller and additional other methods, such as in Ex.~\ref{ex:bball}. 
The method for the lower boundary (\abs{down}) does not establish safety by its recursive call, as the level may rise above the upper boundary.
Similarly, if a method is added to Ex.~\ref{ex:ctank}, it cannot use the information from \abs{ctrl}.
To amend this, we use that in controlled classes the controllers can always be scheduled if their condition holds. 
Their conditions, thus, can be used in a more precise computation of the post-region.

Intuitively, we can use the same structure as in Lem.~\ref{lem:hcp:3}, 
to make use of the fact that controllers are always running: their guards form the very same post-region as the guards of the procedures of concurrent programs.
\begin{theorem}[Structurally Controlled Regions]\label{lem:control}
%
The post-region generator $\Psi^\mathsf{ctrl}$, which maps every method \abs{C.m}, and every identifier $i$ of an \abs{await} with \abs{C.m},
to the following formula,
\[
\bigwedge_{\methodname \in \mathsf{Ctrl}_{\xabs{C}}} \widetilde{\neg}\mathit{ttrig}_{\methodname}
\]
is sound.
\end{theorem}
While this increases precision, robustness suffers: removal of a controller requires reproving all methods in the class.
This is to be expected, as a controller influences every post-region.
\begin{lemma}\label{lem:controlmod}
Let $\xabs{C}$ be a safe class with structurally controlled regions. 
Let $\xabs{C}^-$ be \xabs{C} with some method $\methodname^-$ removed and $\xabs{C}^+$ be \xabs{C} with some method added. 
\begin{enumerate}
\item
If $\methodname^-$ is a controller, then all methods of $\xabs{C}^-$ have to be reproven.
\item 
If $\methodname^-$ is not a controller, then the proof obligations of methods calling $\methodname^-$ must be shown.
\item
To verify $\xabs{C}^+$, only the proof obligation of the new method has to be shown.
\end{enumerate}
\end{lemma}

To demonstrate the use of multiple controllers, we present a model of a ball on a billard table with additional operations and verify that the ball indeed stays within the boundary. As we focus on modularity, we keep the dynamics simple.

\begin{example}
The core model is given in the upper part of Fig.~\ref{fig:billard}.
It has one (non-physical) field (\abs{top}, \abs{bottom}, \abs{left}, \abs{right}) that models the x or y-coordinate of the respective boundary
and one controller (\abs{ctrlTop}, \abs{ctrlBottom}, \abs{ctrlLeft}, \abs{ctrlRight})for each boundary of the table. 
The creation condition specifies that the boundaries have some space in-between.
The position of the ball starts at $(0,0)$ within the boundaries. Afterwards the ball travels with some two dimensional velocity and bounces off the boundaries.
The controller realizes this bounce by inverting the respective component of the velocity.
We can verify the condition that in this model, the ball is indeed always on the table.

Now, we add further operations, as stated in the lower part of Fig.~\ref{fig:billard}. 
One method to accelerate the ball (\abs{acc}), one method to push it in some direction (\abs{push}), one method to push the top boundary (\abs{incSize})
and one method to make the ball leap forward by 1 length unit along the x-axis if there is enough space (\abs{leap}).

Each modification requires us \emph{only} to verify the new method. 
The controllers of the core code are abstracted into the post-region of each added method and must not be reverified as no further controller is added.
The other methods do not call each other and are, thus, safe independent of the presence of other (non-controlling) operations.
Note that we do not specify how often and when each method is called: the system is safe for any number of method calls at arbitrary points in time.
This also entails that the table is arbitrary big, as the top boundary can be pushed back by an arbitrary distance.

\begin{figure}
\begin{abscode}
interface IBillard { }
[HybridSpec: Requires(top > 0 && bottom < 0 && right > 0 && left < 0)]
class Billard(Real top, Real bottom, 
             Real left, Real right,
             Real vx,   Real vy) implements IBillard{

[HybridSpec: ObjInv(   top >= y && bottom <= y 
                    && right >= x && left <= left)]
physical{
  Real x = 0: x' = vx;
  Real y = 0: y' = vy;
}
{ 
  this!ctrlTop(); this!ctrlRight(); 
  this!ctrlBottom(); this!ctrlLeft();
}
Unit ctrlTop(){ 
    await diff y >= top && vy >= 0; vy = -vy; this!ctrlTop(); 
}
Unit ctrlBottom(){ 
    await diff y >= bottom && vy <= 0; vy = -vy; this!ctrlBottom(); 
}
Unit ctrlRight(){ 
    await diff x >= right && vx >= 0; vx = -vx; this!ctrlRight(); 
}
Unit ctrlLeft(){ 
    await diff x >= left && vx <= 0; vx = -vx; this!ctrlLeft(); 
}
}
\end{abscode}
\begin{abscode}
//in the interface
Unit accelerate(Real r);
Unit push(Real a, Real b);
Unit leap();
[HybridSpec: Requires(extra >= 0)]
Unit incSize(Real extra);
//in the class
Unit accelerate(Real r){ vx = vx*r; vy = vy*r; }
Unit push(Real a, Real b){ vx = vx + a; vy = vy + b; }
Unit leap(){ 
    if( vx > 0 && right - x > 1) x = x + 1; 
    if( vx < 0 && x -left > 1) x =x - 1; 
}
Unit incSize(Real extra){ top  = top + extra; }
\end{abscode}
\caption{A billard model.}
\label{fig:billard}
\end{figure}
\end{example}

\subsection{Discussion}\label{ssec:discussion}
There is a different trade-off between precision and robustness for each notion of post-region, where under \emph{precision} we understand the size of the post-region.
The smaller a post-region it, the more precise it is. We remind that under robustness, we understand the reuse of proof obligations if a part of the program is changed.
The more proof obligations can be reused if a method is changed, the more rebust we consider the notion of post-region.
We ignore method pre- and post-conditions, which are exhaustively discussed in literature~\citep{key}, and focus on object invariants.

Changes in the dynamics in the physical block and after changes to the object invariant it is obviously required to reprove anything.
\begin{itemize}
\item
Basic regions are the most imprecise, as their post-region generator always uses the maximal post-region. This notion is, however, the most robust:
any change to a method, requires only to reprove the changed method: In terms of Lem.~\ref{lem:modulartrivial}), a change is a removal followed by an addition of the new variant of the method.
\item 
Locally controlled regions are a middle-ground: they are more precise than post-regions, but less robust.
With this notion, changes in the leading guard of a method require reproving all methods calling it, as it forms a part of their post-region (Lem.~\ref{lem:modulareuniform}).
\item 
Structurally controlled regions are the most precise, as they add additional conjuncts to the evolution domain constraint describing the locally controlled post-region.
They are, in turn, the least robust: a change in the leading guard of any controller requires reproving \emph{every} method in the class (Lem.~\ref{lem:controlmod}).
\end{itemize}
Note that all post-regions described here are robust with respect to changes in other classes.

\paragraph{Interactions.}
Our system ensures that objects of a verified class are safe in every context that adheres to its method contracts and creation conditions -- 
Active Objects enable us to decompose the system into strongly encapsulated parts with local invariants. 
As with other systems for Active Objects~\citep{bpl,DinBH15}, a global property must be broken down into local object invariants outside the verification logic itself either manually, such as in the history-based approaches~\citep{icfem,KamburjanH17}, or via an external specification tool~\citep{DBLP:journals/corr/abs-2208-04630}.
The event-based decomposition~\citep{DinO15} can be transfered directly to \texttt{HABS}, as the events are independent of time.
We next give an example with two objects that shows how objects and method contracts interact even without a decomposition.
\begin{example}
We consider the Lotka-Volterra equations for two patches and migration from \cite{jansen}, where the migration is discrete and conditional:
prey migrates from one patch to another only if it outnumbers the predators. 
The interface of a patch is given in Fig.~\ref{fig:lv:interface}.
It declares a method to connect to another patch (\abs{setOther}) and one for migrate of \abs{n} prey to it (\abs{to}). For simplicity, we allow
the migration of non-integer prey, which is also possible in the continuous migration model~\citep{jansen}.
\begin{figure}
\begin{abscode}
interface Patch{
  [HybridSpec: Requires(n >= 0)]
  Unit to(Real n);
  Unit setOther(Patch nOther)
}
\end{abscode}
\caption{Interface of a patch.}
\label{fig:lv:interface}
\end{figure}

The implementation of a patch is given in Fig.~\ref{fig:lv:class}.
It declares the initial prey (\abs{ix}) and predators (\abs{iy}), and the four parameters to control the rates of population change (\abs{alpha}, \abs{beta}, \abs{gamma}, \abs{delta}).
The \abs{physical} block declares the usual Lotka-Volterra equations and the safety invariant states that the size of the prey and predator populations remains strictly positive.
The \abs{migrate} method checks once per time unit if the prey population is at least ten times larger than the predator population and, in this case,
migrates 10\% to the other patch.
The other methods have their obvious implementation.
\begin{figure}
\begin{abscode}
[HybridSpec: Requires(ix >= 0 & iy >= 0 & alpha >= 0 & beta >= 0 
                      & gamma >= 0 & delta >= 0)]
[HybridSpec: ObjInv(x >= 0 & y >= 0 & alpha >= 0 & beta >= 0 
                    & gamma >= 0 & delta >= 0)]
class Patch(Real ix, Real iy, 
            Real alpha, Real beta, Real delta, Real gamma, 
            Patch other) 
  implements Patch {
  physical{
    Real x = ix : x' = alpha*x - beta*x*y;
    Real y = iy : y' = delta*y*x - gamma*y;
  }
  { this!migrate(); }

  Unit migrate(){
    await duration(1,1);
    if( x >= 10*y ){
      other!toMig(this.x/10);
      this.x = this.x*(9/10);
    }
    this!migrate();
  }

  Unit to(Real n){ this.x = this.x + n; }
  Unit setOther(Patch nOther){ this.other = nOther; }
}
\end{abscode}
\caption{Implementation of a patch.}
\label{fig:lv:class}
\end{figure}

The main block to set up the two patches is given in Fig.~\ref{fig:lv:main}.
\begin{figure}
\begin{abscode}
{
  Patch p1 = new Patch(100, 10, ..., null);
  Patch p2 = new Patch(100, 10, ..., null);
  p1!setOther(p2);
  p2!setOther(p2);
}
\end{abscode}
\caption{Setting up two patches.}
\label{fig:lv:main}
\end{figure}

The example can be verified using basic post-regions and illustrates how the interdependence of invariant and contract:
That there are no negative animals due to the discrete interactions (as stated by the invariant), 
depends on the discrete method contract of \abs{to} -- if the precondition is removed, verification of \abs{to} fails.
The that the pre-condition is actually upheld by the caller depends on the invariant -- if the object invariant is removed, verification of \abs{migrate} fails.
\end{example}

\section{Implementation and Integration}\label{sec:multi}
We implemented the procedures in Sec.~\ref{sec:region} in the \texttt{Chisel} tool\footnote{The source code is available under \texttt{https://github.com/Edkamb/chisel-tool/}. For the convenience of the reviewers, a VM is available under \texttt{https://ln.sync.com/dl/9716fcfb0/5aw2yw3k-6hiacyg9-ikz6wdzr-4dhqnytt}.}.
\texttt{Chisel} reuses the \texttt{HABS} parser.
The tool parses the input file and checks which class is controlled.
The user can select to verify all classes, only one class, only one method or only the main block. 
Additionally, the user may select which approach for region computation is to be used.
It then translates a class into a set of \texttt{KeYmaera X} proof obligations and automatically tries to verify them.
Failed proof attempts can be retrieved by the user and \texttt{Chisel} allows annotating proof scripts to manually apply a tactic instead of the automatic procedures:
Such tactics are output by \texttt{KeYmaera X} after a proof is closed.

\noindent
\begin{abscode}[numbers=none]
[HybridSpec: Tactic("expandAllDefs; master; DEs(1)")]
Unit method(){ ... }
\end{abscode}

All \HABS-examples in this work are uploaded with the source code of the tool and can be closed fully automatically. 
Reverification of the water tank from \cite{lites} shows that the new approach is more suited for automated proving:
The previous system, which we describe in more detail in Sec.~\ref{sec:related}, required interaction to close Ex.~\ref{ex:bball}~\citep[Fig.~2]{arxiv} but our new tool \texttt{Chisel} can verify it fully automatically due to more simple proof obligations.
The reason is that the concurrency model and scheduling is now implicit for the prover, while it was encoded directly in the $d\mathcal{L}$ formula before and required to retrieve two additional loop invariants.
While the previous system only handled controlled classes where \abs{await} is only allowed as the first statement of a method, \texttt{Chisel} covers the complete language (except additional data types and exceptions). 

Concerning the complexity of dynamics, we point out that our system for \texttt{HABS} depends completely on the $d\mathcal{L}$ solver to handle the dynamics -- 
in practice, \texttt{Chisel} handles the same complexity of dynamics as \texttt{KeYmaera X}, which it uses to discharge the proof obligations.

\paragraph{Data-Rich Modeling.}
\texttt{Chisel} can be used in combination with \texttt{Crowbar}~\citep{crowbar}, the verification tool for \texttt{ABS}. 
While \texttt{Chisel} verifies hybrid classes, \texttt{Crowbar} verifies non-hybrid classes. BPL and \texttt{Crowbar} support data-rich modeling: arbitrary data types and a functional language that operates on them, additionally to more expressive expressions.
Both \texttt{Chisel} and \texttt{Crowbar} tools support method contracts as \emph{cooperative contracts}~\citep{doa} and interactions are handled automatically by this mechanism.
\texttt{Crowbar} uses \abs{[Spec: ...]} annotations instead of \abs{[HybridSpec: ...]}, so a method serving as the interface must be annotated with both (using the same formula).
If data-rich modeling in \texttt{ABS} and hybrid modeling in \texttt{HABS} can be modularized into different classes that only communicate via methods calls, 
then we can verify hybrid systems with complex data types.

\section{Related Work}\label{sec:related}
For a detailed survey of modularity in verification we refer to \cite{GurovHK20}.
\paragraph{Hybrid Programming Languages.}
There are only few programming languages for hybrid systems with full-fledged formal semantics.
The best examined one is the algebraic language of \ddl introduced by \cite{DBLP:conf/lics/Platzer12b,DBLP:journals/jar/Platzer17,Platzer18}.
It is a low level language without the primitives for concurrency used in HAOs.
It has been recently extended to include CSP-style operators for concurrency by \cite{DBLP:journals/corr/abs-2303-17333}. In contrast to post-regions, which use the inbuilt encapsulation and concurrency of Active Objects, their approach focuses on decomposing low-level operators into standard \ddl through special rules.

$\mathtt{While}^\mathtt{dt}$~\citep{hybridwhileold} is a hybrid while-language with a verification condition generator~\citep{DBLP:conf/cav/HasuoS12}, but neither simulation nor compositionality is examined. Semantically, it is based on non-standard analysis. The language is sequential.

\texttt{HybCore}~\citep{hybridwhile} is another hybrid while-language. It has a formal semantics and a simulator~\citep{ictac}, but verification is not considered. Its SOS semantics addresses the inadequacy of \ddl in case of non-terminating loops, but does not offer additional structuring elements. 
The language is sequential and less expressive than \ddl due to a lack of evolution constraints, which are needed for event-based modeling.

Hybrid Rebeca~\citep{hybreb} embeds hybrid automata directly into the actor language Rebeca. 
Discrete objects and hybrid automata can communicate via messages, but are distinct entities.
It has no semantics in first principles, but translates into hybrid automata. 
It is worth noting that despite a class structure and an actor-like concurrency model, its verification is non-modular and requires to monolithically verify a single automaton for the whole program.

\HABS has a predecessor system to verify invariants~\citep{arxiv,lites},
However, it uses one proof obligation \emph{per class} and has several flaws that inhibit one from applying it to more realistic systems: 
(1) after changes in one method the whole class must be re-proven,
(2) the size proof obligation is exponential in the number of methods, and
(3) two\footnote{One loop advances time, one loop allows multiple methods to run without advancing time.} more loop invariants must be inferred, which hampers automatization.

Process algebras are minimalistic programming languages that have spawned several formalisms for distributed hybrid systems. None of them has been considered for verification and we, thus, only refer to them for completeness's sake and refrain from a detailed discussion. The \texttt{CCPS} system~\citep{LanotteM17} is an extension of timed process algebra TPL~\citep{HennessyR95} and CCS~\citep{Milner80} and uses an inbuilt notion of sensor and actuators.
The $\varphi$-calculus~\citep{RoundsS03} is an extension of the $\pi$-calculus. It has no physical processes but considers them as a part of the environment. 
The work of \cite{Khadim} gives a detailed comparison on the process algebras \texttt{HyPA}~\citep{CuijpersR05}, \texttt{Hybrid $\chi$}~\citep{BosK03}, both extending ACP~\citep{BergstraK85}, the $\varphi$-calculus and another extension of ACP~\citep{BergstraM05}.
The \texttt{HYPE} calculus~\citep{GalpinBH13} is an approach that focuses on the composition of continuous behavior, less so on the interaction through discrete actions.

\paragraph{Modularity in Deductive Verification of Hybrid Systems}
The minimal structure of the language underlying \ddl requires that compositionality to be encoded in 
elaborate proof structures~\citep{DBLP:conf/acsd/LunelBT17,DBLP:conf/fm/LunelMBT19,DBLP:journals/sttt/MullerMRSP18}.
Modularity is, thus, mainly a decomposition of subproofs in the sequent calculus and follows the structure of the overall formula, not the structure of the program.

The \texttt{HHL} prover of \cite{hhl} is based on Hoare triples over hybrid CSP. 
Its implementation is based on an embedding of the triples into Isabelle/HOL, which is then used to prove validity.
As such it has also access to the proof structuring techniques of Isabelle, such as lemmata.

\cite{baar} give a system to modularize \ddl proofs by transforming hybrid programs into control-flow graphs and annotating additional contracts to the edges.
The procedure results in smaller proof obligations for each edge, compared to the overall proof goal.
The resulting structure is similar to a hybrid automaton, but the connection is not formally established and reachability approaches are not considered.

%
%

\paragraph{Analysis of Hybrid Programming Languages}

The Zelus language~\citep{HSCC2013} employs a static analysis to ensures certain safety properties for its simulation, for example to avoid that discontinuities occur during integration.
It implements parallelism through its synchronous concurrency model~\citep{DBLP:journals/pieee/BenvenisteB91}. The language has no verification system.

\paragraph{Other Related Approaches}
For an overview over contracts for cyber-physical systems from a system engineering perspective we refer to \cite{BenvenisteCNPRR18}.

Finally, we discuss some semantic compositions.
For hybrid automata, composition of \emph{semantics} has been investigated since the beginning by parallel composition~\citep{alur} or by marking some variable as input or output ports, e.g., in hybrid I/O automata~\citep{lynch}, which can be used to build a hierarchical model made of components~\citep{donze} for model checkers.
Composition has also been investigated for model checking hybrid automata by assume-guarantee reasoning to decompose systems~\citep{minea,frehse1}. 
Such approaches are similar to method contracts, but rely-guarantee reasoning for hybrid automata does not encapsulate of data and method calls, but abstracts of components executed in parallel. 

The above approaches share with ours that the continuous behavior is not composed. \cite{CuijpersR05} consider an algebraic approach to also synchronize on continuous behavior.
Other process algebraic approaches define further composition operators for models, such as Hybrid CPS~\citep{hcsp} or the $\varphi$-calculus, a hybrid $\pi$-calculus~\citep{RoundsS03}. 
None of these approaches investigate how this composition can be used for deductive verification of functional properties, a more detailed comparison between them is given in~\citep{BergstraM05}.

Region automata~\citep{AlurD94} are a namesake of our system that compute regions in the state space of timed automata to decide reachability.
Beyond the name, the systems share only the idea to split the state space into regions to simplify analysis.


\section{Conclusion}\label{sec:conclusion}
This is the first work to successfully generalize specification and verification principles for object-oriented languages to a hybrid setting.
Instead of verifying an object invariant by using it as a post-condition of a method, we instead use it as an invariant for the post-region: 
the states reachable from the post-state method when following the given dynamics.
Verification of hybrid systems is hard, but we show that by using an object-oriented programming languages as a host for hybrid behavior, 
it is possible to use the additional structure provided by the language to improve robustness of verification.

\paragraph{Future Work.} 
This work describes the verification of invariants local to single objects, which relies on the strong encapsulation of (Hybrid) Active Objects.
Decomposition of global invariants of the whole system to local ones remains an active research area even for non-times Active Objects,
and has so far been performed either manually~\citep{icfem,KamburjanH17,DBLP:conf/birthday/BoerG20} or only for a restricted setting~\citep{DBLP:journals/corr/abs-2208-04630}.
Trace logics and semantics that conceptually separate local from global behavior~\citep{DBLP:journals/corr/abs-2202-12195,DBLP:conf/tableaux/DinHJPT17,bpl} show promise as a research direction here. We plan to extend the recent trace semantics for Timed ABS~\citep{DBLP:conf/birthday/Tarifa20} to connect \ddl more closely with BPL and verify global invariants.

\subsubsection*{Acknowledgments}
This work was supported by the Norwegian Research Council via the \texttt{PeTWIN} project (Grant Nr.\ 294600) and the SIRIUS research center (Grant Nr.\ 237898).

\bibliography{ref}

\begin{thebibliography}{68}
\providecommand{\natexlab}[1]{#1}
\providecommand{\url}[1]{{#1}}
\providecommand{\urlprefix}{URL }
\providecommand{\doi}[1]{\url{https://doi.org/#1}}
\providecommand{\eprint}[2][]{\url{#2}}
 \bibcommenthead

\bibitem[{Ahrendt et~al(2016)Ahrendt, Beckert, Bubel, H{\"{a}}hnle, Schmitt,
  and Ulbrich}]{key}
Ahrendt W, Beckert B, Bubel R, et~al (eds)  (2016) Deductive Software
  Verification - The KeY Book - From Theory to Practice, {LNCS}, vol 10001.
  Springer, \doi{10.1007/978-3-319-49812-6}

\bibitem[{Albert et~al(2014)Albert, de~Boer, H{\"{a}}hnle, Johnsen, Schlatte,
  {Tapia Tarifa}, and Wong}]{AlbertBHJSTW14}
Albert E, de~Boer FS, H{\"{a}}hnle R, et~al (2014) Formal modeling and analysis
  of resource management for cloud architectures: an industrial case study
  using real-time {ABS}. Service Oriented Computing and Applications
  8(4):323--339

\bibitem[{Alur and Dill(1994)}]{AlurD94}
Alur R, Dill DL (1994) A theory of timed automata. Theor Comput Sci
  126(2):183--235. \doi{10.1016/0304-3975(94)90010-8},
  \urlprefix\url{https://doi.org/10.1016/0304-3975(94)90010-8}

\bibitem[{Alur et~al(1995)Alur, Courcoubetis, Halbwachs, Henzinger, Ho,
  Nicollin, Olivero, Sifakis, and Yovine}]{alur}
Alur R, Courcoubetis C, Halbwachs N, et~al (1995) The algorithmic analysis of
  hybrid systems. Theor Comput Sci 138(1):3--34.
  \doi{10.1016/0304-3975(94)00202-T}

\bibitem[{Baar and Staroletov(2019)}]{baar}
Baar T, Staroletov S (2019) A control flow graph based approach to make the
  verification of cyber-physical systems using {KeYmaera} easier. Modeling and
  Analysis of Information Systems 25(5):465--480.
  \doi{10.18255/1818-1015-465-480}

\bibitem[{Benveniste and Berry(1991)}]{DBLP:journals/pieee/BenvenisteB91}
Benveniste A, Berry G (1991) The synchronous approach to reactive and real-time
  systems. Proc {IEEE} 79(9):1270--1282. \doi{10.1109/5.97297},
  \urlprefix\url{https://doi.org/10.1109/5.97297}

\bibitem[{Benveniste et~al(2018)Benveniste, Caillaud, Nickovic, Passerone,
  Raclet, Reinkemeier, Sangiovanni{-}Vincentelli, Damm, Henzinger, and
  Larsen}]{BenvenisteCNPRR18}
Benveniste A, Caillaud B, Nickovic D, et~al (2018) Contracts for system design.
  Found Trends Electron Des Autom 12(2-3):124--400. \doi{10.1561/1000000053}

\bibitem[{Bergstra and Klop(1985)}]{BergstraK85}
Bergstra JA, Klop JW (1985) Algebra of communicating processes with
  abstraction. Theor Comput Sci 37:77--121

\bibitem[{Bergstra and Middelburg(2005)}]{BergstraM05}
Bergstra JA, Middelburg CA (2005) Process algebra for hybrid systems. Theor
  Comput Sci 335(2-3):215--280. \doi{10.1016/j.tcs.2004.04.019}

\bibitem[{Bezirgiannis et~al(2019)Bezirgiannis, de~Boer, Johnsen, Pun, and
  {Tapia Tarifa}}]{BezirgiannisBJP19}
Bezirgiannis N, de~Boer FS, Johnsen EB, et~al (2019) Implementing {SOS} with
  active objects: {A} case study of a multicore memory system. In: H{\"{a}}hnle
  R, van~der Aalst WMP (eds) {FASE}, {LNCS}, vol 11424. Springer, pp 332--350,
  \doi{10.1007/978-3-030-16722-6\_20}

\bibitem[{Bj{\o}rk et~al(2013)Bj{\o}rk, de~Boer, Johnsen, Schlatte, and {Tapia
  Tarifa}}]{BjorkBJST13}
Bj{\o}rk J, de~Boer FS, Johnsen EB, et~al (2013) User-defined schedulers for
  real-time concurrent objects. Innovations in Systems and Software Engineering
  9(1):29--43

\bibitem[{de~Boer and de~Gouw(2022)}]{DBLP:conf/birthday/BoerG20}
de~Boer FS, de~Gouw S (2022) Reasoning about active objects: {A} sound and
  complete assertional proof method. In: The Logic of Software. {A} Tasting
  Menu of Formal Methods, Lecture Notes in Computer Science, vol 13360.
  Springer, pp 173--192

\bibitem[{de~Boer et~al(2017)de~Boer, Serbanescu, H{\"{a}}hnle, Henrio, Rochas,
  Din, Johnsen, Sirjani, Khamespanah, Fernandez{-}Reyes, and
  Yang}]{BoerSHHRDJSKFY17}
de~Boer FS, Serbanescu V, H{\"{a}}hnle R, et~al (2017) A survey of active
  object languages. {ACM} Computing Surveys 50(5):1--39

\bibitem[{Bos and Kleijn(2003)}]{BosK03}
Bos V, Kleijn JJT (2003) Redesign of a systems engineering language:
  Formalisation of {X}. Formal Aspects Comput 15(4):370--389

\bibitem[{Bourke and Pouzet(2013)}]{HSCC2013}
Bourke T, Pouzet M (2013) Z\`elus: A synchronous language with {ODEs}. In: 16th
  International Conference on Hybrid Systems: Computation and Control
  (HSCC'13), Philadelphia, USA, pp 113--118

\bibitem[{Brieger et~al(2023)Brieger, Mitsch, and
  Platzer}]{DBLP:journals/corr/abs-2303-17333}
Brieger M, Mitsch S, Platzer A (2023) Uniform substitution for dynamic logic
  with communicating hybrid programs. CoRR abs/2303.17333

\bibitem[{Cuijpers and Reniers(2005)}]{CuijpersR05}
Cuijpers PJL, Reniers MA (2005) Hybrid process algebra. J Log Algebraic Methods
  Program 62(2):191--245

\bibitem[{Din and Owe(2015)}]{DinO15}
Din CC, Owe O (2015) Compositional reasoning about active objects with shared
  futures. Formal Asp Comput 27(3):551--572. \doi{10.1007/s00165-014-0322-y}

\bibitem[{Din et~al(2015)Din, Bubel, and H{\"a}hnle}]{DinBH15}
Din CC, Bubel R, H{\"a}hnle R (2015) {KeY-ABS}: A deductive verification tool
  for the concurrent modelling language {ABS}. In: Felty A, Middeldorp A (eds)
  {CADE}, LNCS, vol 9195. Springer, pp 517--526,
  \doi{10.1007/978-3-319-21401-6\_35}

\bibitem[{Din et~al(2017)Din, H{\"{a}}hnle, Johnsen, Pun, and
  Tarifa}]{DBLP:conf/tableaux/DinHJPT17}
Din CC, H{\"{a}}hnle R, Johnsen EB, et~al (2017) Locally abstract, globally
  concrete semantics of concurrent programming languages. In: {TABLEAUX},
  Lecture Notes in Computer Science, vol 10501. Springer, pp 22--43

\bibitem[{Din et~al(2022)Din, H{\"{a}}hnle, Henrio, Johnsen, Pun, and
  Tarifa}]{DBLP:journals/corr/abs-2202-12195}
Din CC, H{\"{a}}hnle R, Henrio L, et~al (2022) {LAGC} semantics of concurrent
  programming languages. CoRR abs/2202.12195

\bibitem[{Donz{\'{e}} and Frehse(2013)}]{donze}
Donz{\'{e}} A, Frehse G (2013) Modular, hierarchical models of control systems
  in {SpaceEx}. In: {ECC} 2013. {IEEE}, pp 4244--4251

\bibitem[{Flanagan and Felleisen(1999)}]{DBLP:journals/jfp/FlanaganF99}
Flanagan C, Felleisen M (1999) The semantics of future and an application. J
  Funct Program 9(1):1--31. \doi{10.1017/s0956796899003329},
  \urlprefix\url{https://doi.org/10.1017/s0956796899003329}

\bibitem[{{Frehse} et~al(2004){Frehse}, {Zhi Han}, and {Krogh}}]{frehse1}
{Frehse} G, {Zhi Han}, {Krogh} B (2004) Assume-guarantee reasoning for hybrid
  i/o-automata by over-approximation of continuous interaction. In: {CDC} 2004,
  pp 479--484 Vol.1, \doi{10.1109/CDC.2004.1428676}

\bibitem[{Fulton et~al(2015)Fulton, Mitsch, Quesel, V{\"{o}}lp, and
  Platzer}]{FultonMQVP15}
Fulton N, Mitsch S, Quesel J, et~al (2015) {KeYmaera} {X:} an axiomatic
  tactical theorem prover for hybrid systems. In: {CADE}, {LNCS}, vol 9195.
  Springer, pp 527--538, \doi{10.1007/978-3-319-21401-6\_36}

\bibitem[{Galpin et~al(2013)Galpin, Bortolussi, and Hillston}]{GalpinBH13}
Galpin V, Bortolussi L, Hillston J (2013) {HYPE:} hybrid modelling by
  composition of flows. Formal Aspects Comput 25(4):503--541

\bibitem[{Goncharov and Neves(2019)}]{hybridwhile}
Goncharov S, Neves R (2019) An adequate while-language for hybrid computation.
  In: Komendantskaya E (ed) {PPDP}. {ACM}, pp 11:1--11:15,
  \doi{10.1145/3354166.3354176}

\bibitem[{Goncharov et~al(2020)Goncharov, Neves, and Proen{\c{c}}a}]{ictac}
Goncharov S, Neves R, Proen{\c{c}}a J (2020) Implementing hybrid semantics:
  From functional to imperative. In: Pun VKI, Stolz V, Sim{\~{a}}o A (eds)
  {ICTAC}, LNCS, vol 12545. Springer, pp 262--282,
  \doi{10.1007/978-3-030-64276-1\_14}

\bibitem[{Grahl et~al(2016)Grahl, Bubel, Mostowski, Schmitt, Ulbrich, and
  Wei{\ss}}]{bubel}
Grahl D, Bubel R, Mostowski W, et~al (2016) Modular specification and
  verification. In: The KeY Book, {LNCS}, vol 10001. Springer,
  \doi{10.1007/978-3-319-49812-6\_9}

\bibitem[{Gurov et~al(2020)Gurov, H{\"{a}}hnle, and Kamburjan}]{GurovHK20}
Gurov D, H{\"{a}}hnle R, Kamburjan E (2020) Who carries the burden of
  modularity? - {Introduction} to {ISoLA} 2020 track on modularity and
  (de-)composition in verification. In: ISoLA {(1)}, LNCS, vol 12476. Springer,
  pp 3--21, \doi{10.1007/978-3-030-61362-4\_1}

\bibitem[{{Halstead Jr.}(1985)}]{DBLP:journals/toplas/Halstead85}
{Halstead Jr.} RH (1985) Multilisp: {A} language for concurrent symbolic
  computation. {ACM} Trans Program Lang Syst 7(4):501--538.
  \doi{10.1145/4472.4478}, \urlprefix\url{https://doi.org/10.1145/4472.4478}

\bibitem[{Harel et~al(2000)Harel, Tiuryn, and Kozen}]{dynamic}
Harel D, Tiuryn J, Kozen D (2000) Dynamic Logic. MIT Press

\bibitem[{Hasuo and Suenaga(2012)}]{DBLP:conf/cav/HasuoS12}
Hasuo I, Suenaga K (2012) Exercises in nonstandard static analysis of hybrid
  systems. In: Madhusudan P, Seshia SA (eds) {CAV}, {LNCS}, vol 7358. Springer,
  pp 462--478, \doi{10.1007/978-3-642-31424-7\_34}

\bibitem[{Hennessy and Regan(1995)}]{HennessyR95}
Hennessy M, Regan T (1995) A process algebra for timed systems. Inf Comput
  117(2):221--239

\bibitem[{Henzinger et~al(2001)Henzinger, Minea, and Prabhu}]{minea}
Henzinger TA, Minea M, Prabhu VS (2001) Assume-guarantee reasoning for
  hierarchical hybrid systems. In: Benedetto MDD, Sangiovanni{-}Vincentelli AL
  (eds) {HSCC}, {LNCS}, vol 2034. Springer, pp 275--290,
  \doi{10.1007/3-540-45351-2\_24}

\bibitem[{Hoare(1969)}]{DBLP:journals/cacm/Hoare69}
Hoare CAR (1969) An axiomatic basis for computer programming. Commun {ACM}
  12(10):576--580. \doi{10.1145/363235.363259}

\bibitem[{Jahandideh et~al(2018)Jahandideh, Ghassemi, and Sirjani}]{hybreb}
Jahandideh I, Ghassemi F, Sirjani M (2018) Hybrid rebeca: Modeling and
  analyzing of cyber-physical systems. In: Chamberlain RD, Taha W,
  T{\"{o}}rngren M (eds) {CyPhy}, {LNCS}, vol 11615. Springer, pp 3--27,
  \doi{10.1007/978-3-030-23703-5\_1}

\bibitem[{Jansen(1995)}]{jansen}
Jansen V (1995) Regulation of predator-prey systems through spatial
  interactions: A possible solution to the paradox of enrichment. Oikos
  74:384--390. \doi{10.2307/3545983}

\bibitem[{Jifeng(1994)}]{hcsp}
Jifeng H (1994) From CSP to Hybrid Systems, Prentice Hall International (UK)
  Ltd., p 171–189

\bibitem[{Johnsen et~al(2010)Johnsen, H{\"{a}}hnle, Sch{\"{a}}fer, Schlatte,
  and Steffen}]{JohnsenHSSS10}
Johnsen EB, H{\"{a}}hnle R, Sch{\"{a}}fer J, et~al (2010) {ABS:} {A} core
  language for abstract behavioral specification. In: Aichernig BK, de~Boer FS,
  Bonsangue MM (eds) {FMCO}'10, LNCS, vol 6957. Springer, pp 142--164,
  \doi{10.1007/978-3-642-25271-6\_8}

\bibitem[{Kamburjan(2019)}]{bpl}
Kamburjan E (2019) Behavioral program logic. In: Cerrito S, Popescu A (eds)
  {TABLEAUX}, {LNCS}, vol 11714. Springer, pp 391--408,
  \doi{10.1007/978-3-030-29026-9\_22}

\bibitem[{Kamburjan(2021)}]{HSCC}
Kamburjan E (2021) From post-conditions to post-region invariants: Deductive
  verification of hybrid objects. In: {HSCC}, in print, preprint is part of the
  uploaded VM. ACM

\bibitem[{Kamburjan and H{\"{a}}hnle(2017)}]{KamburjanH17}
Kamburjan E, H{\"{a}}hnle R (2017) Deductive verification of railway
  operations. In: Fantechi A, Lecomte T, Romanovsky AB (eds) {RSSRail}, {LNCS},
  vol 10598. Springer, pp 131--147, \doi{10.1007/978-3-319-68499-4\_9}

\bibitem[{Kamburjan and Wasser(2022)}]{DBLP:journals/corr/abs-2208-04630}
Kamburjan E, Wasser N (2022) The right kind of non-determinism: Using
  concurrency to verify {C} programs with underspecified semantics. In: {ICE},
  pp 1--16

\bibitem[{Kamburjan et~al(2016)Kamburjan, Din, and Chen}]{icfem}
Kamburjan E, Din CC, Chen T (2016) Session-based compositional analysis for
  actor-based languages using futures. In: Ogata K, Lawford M, Liu S (eds)
  {ICFEM}, pp 296--312, \doi{10.1007/978-3-319-47846-3\_19}

\bibitem[{Kamburjan et~al(2018)Kamburjan, H{\"{a}}hnle, and Sch{\"{o}}n}]{scp}
Kamburjan E, H{\"{a}}hnle R, Sch{\"{o}}n S (2018) Formal modeling and analysis
  of railway operations with active objects. Sci Comput Program 166:167--193.
  \doi{10.1016/j.scico.2018.07.001}

\bibitem[{Kamburjan et~al(2019)Kamburjan, Mitsch, Kettenbach, and
  H{\"{a}}hnle}]{arxiv}
Kamburjan E, Mitsch S, Kettenbach M, et~al (2019) Modeling and verifying
  cyber-physical systems with hybrid active objects. CoRR abs/1906.05704

\bibitem[{Kamburjan et~al(2020)Kamburjan, Din, H{\"{a}}hnle, and Johnsen}]{doa}
Kamburjan E, Din CC, H{\"{a}}hnle R, et~al (2020) Behavioral contracts for
  cooperative scheduling. In: Ahrendt W, Beckert B, Bubel R, et~al (eds)
  Deductive Verification: The State of the Future, LNCS, vol 12345. Springer,
  \doi{10.1007/978-3-030-64354-6\_4}

\bibitem[{Kamburjan et~al(2022)Kamburjan, Mitsch, and H{\"{a}}hnle}]{lites}
Kamburjan E, Mitsch S, H{\"{a}}hnle R (2022) A hybrid programming language for
  formal modeling and verification of hybrid systems. Leibniz Trans Embed Syst
  In print

\bibitem[{Kamburjan et~al(2023)Kamburjan, Scaletta, and Rollshausen}]{crowbar}
Kamburjan E, Scaletta M, Rollshausen N (2023) Deductive verification of active
  objects with crowbar. Sci Comput Program 226:102928

\bibitem[{Khadim(2006)}]{Khadim}
Khadim U (2006) A comparative study of process algebras for hybrid systems.
  Computer science reports, Technische Universiteit Eindhoven

\bibitem[{Lanotte and Merro(2017)}]{LanotteM17}
Lanotte R, Merro M (2017) A calculus of cyber-physical systems. In: {LATA}, pp
  115--127

\bibitem[{Lin et~al(2020)Lin, Lee, Yu, and Johnsen}]{LinLYJ20}
Lin J, Lee M, Yu IC, et~al (2020) A configurable and executable model of spark
  streaming on apache {YARN}. Int J Grid Util Comput 11(2):185--195.
  \doi{10.1504/IJGUC.2020.105531}

\bibitem[{Lunel et~al(2017)Lunel, Boyer, and Talpin}]{DBLP:conf/acsd/LunelBT17}
Lunel S, Boyer B, Talpin J (2017) Compositional proofs in differential dynamic
  logic {dL}. In: {ACSD}. {IEEE} Computer Society, pp 19--28

\bibitem[{Lunel et~al(2019)Lunel, Mitsch, Boyer, and
  Talpin}]{DBLP:conf/fm/LunelMBT19}
Lunel S, Mitsch S, Boyer B, et~al (2019) Parallel composition and modular
  verification of computer controlled systems in differential dynamic logic.
  In: ter Beek MH, McIver A, Oliveira JN (eds) {FM}, LNCS, vol 11800. Springer,
  pp 354--370, \doi{10.1007/978-3-030-30942-8\_22}

\bibitem[{Lynch et~al(2003)Lynch, Segala, and Vaandrager}]{lynch}
Lynch NA, Segala R, Vaandrager FW (2003) Hybrid {I/O} automata. Inf Comput
  185(1):105--157

\bibitem[{Meyer(1992)}]{meyer}
Meyer B (1992) Applying ``design by contract''. IEEE Computer 25(10):40--51

\bibitem[{Milner(1980)}]{Milner80}
Milner R (1980) A Calculus of Communicating Systems, LNCS, vol~92. Springer

\bibitem[{M{\"{u}}ller et~al(2018)M{\"{u}}ller, Mitsch, Retschitzegger,
  Schwinger, and Platzer}]{DBLP:journals/sttt/MullerMRSP18}
M{\"{u}}ller A, Mitsch S, Retschitzegger W, et~al (2018) Tactical contract
  composition for hybrid system component verification. STTT 20(6):615--643

\bibitem[{Platzer(2010)}]{Platzer10}
Platzer A (2010) Differential-algebraic dynamic logic for
  differential-algebraic programs. J of Logic and Computation 20(1):309--352

\bibitem[{Platzer(2012{\natexlab{a}})}]{QDL}
Platzer A (2012{\natexlab{a}}) A complete axiomatization of quantified
  differential dynamic logic for distributed hybrid systems. {LMCS} 8(4).
  \doi{10.2168/LMCS-8(4:17)2012}

\bibitem[{Platzer(2012{\natexlab{b}})}]{DBLP:conf/lics/Platzer12b}
Platzer A (2012{\natexlab{b}}) The complete proof theory of hybrid systems. In:
  LICS. IEEE, pp 541--550

\bibitem[{Platzer(2017)}]{DBLP:journals/jar/Platzer17}
Platzer A (2017) A complete uniform substitution calculus for differential
  dynamic logic. J Automated Reasoning 59(2):219--265

\bibitem[{Platzer(2018)}]{Platzer18}
Platzer A (2018) Logical Foundations of Cyber-Physical Systems. Springer

\bibitem[{Rounds and Song(2003)}]{RoundsS03}
Rounds WC, Song H (2003) The phi-calculus: {A} language for distributed control
  of reconfigurable embedded systems. In: {HSCC}, LNCS, vol 2623. Springer, pp
  435--449, \doi{10.1007/3-540-36580-X\_32}

\bibitem[{Suenaga and Hasuo(2011)}]{hybridwhileold}
Suenaga K, Hasuo I (2011) Programming with infinitesimals: {A} while-language
  for hybrid system modeling. In: {ICALP} {(2)}, {LNCS}, vol 6756. Springer, pp
  392--403, \doi{10.1007/978-3-642-22012-8\_31}

\bibitem[{Tarifa(2022)}]{DBLP:conf/birthday/Tarifa20}
Tarifa SLT (2022) Locally abstract globally concrete semantics of time and
  resource aware active objects. In: The Logic of Software. {A} Tasting Menu of
  Formal Methods, Lecture Notes in Computer Science, vol 13360. Springer, pp
  481--499

\bibitem[{Wang et~al(2015)Wang, Zhan, and Zou}]{hhl}
Wang S, Zhan N, Zou L (2015) An improved {HHL} prover: An interactive theorem
  prover for hybrid systems. In: {ICFEM}. Springer, pp 382--399,
  \doi{10.1007/978-3-319-25423-4\_25}

\end{thebibliography}



\setcounter{tocdepth}{3}
\setcounter{secnumdepth}{3}
\appendix
\section{Proofs}
\subsection{Proof for Theorem~\ref{thm:main}}
We first show the following lemma, which encapsulates single processes.
Intuitively, it states that if we fix a single process and assume that every other process behaves correctly (1-3), then the fixed process also behaves correctly (a-e). 

\begin{lemma}\label{lem:compose}
Let $\iota^\Psi$ be a sound proof obligation scheme. 
Let $\mathtt{prgm} \in \mathbf{Prgm}$ be a program, $r$ a run of $\mathtt{prgm}$ and
\fid a future of method \abs{C.m} in $r$, running on some object $o$. 
If (1) at the initial state of \fid both the precondition of \abs{C.m} and the invariant of \abs{C} hold in $o$, (2) at every reactivation of \fid the invariant of \abs{C} holds in $o$, and (3) $\iota^\Psi_\xabs{C.m}$ is valid, then (a) the postcondition of \abs{C.m} holds in the final state of \fid, (b) the invariant of \abs{C} holds in $o$ in the final state, (c) the invariant of \abs{C} holds in $o$ until the next process on $o$ is scheduled, including the first state of the next process, (d) the precondition of all called method holds and (e) whenever time advances while \fid is active, the invariant of \abs{C} holds in $o$.
\end{lemma}
\begin{proof}
Let $\tau$ be the combined first state of \fid enriched with \texttt{t} and  \texttt{cll}, both set to zero. 
Note that $\tau$ is also a \ddl-model. Now, as the proof obligation is valid, every state reachable
from $\tau$ via the translation of $\xabs{s}_\xabs{m}$ fulfills the post-condition. 
If every final state of \fid corresponds to some such reachable state, then (a) and (b) hold immediately. As $\Psi$ is sound, every state reachable within the post-region after any such state is a model for the invariant as well -- this is exactly statement (c). It, thus, remains to show that (f) every final state of \fid corresponds to some such reachable state of $\mathsf{trans}(\xabs{s}_\xabs{m})$, as well as (d) and (e).

To do so, we define $\hat\Theta = (\hat\theta_i)_{i\in I}$ as the family of subtraces of $\theta(o)$ such that each $\theta_i$ starts with a configuration where \fid is scheduled and ends with a configuration where \fid is descheduled. We take $I \subseteq \mathbb{N}$ and order the subtraces according to their order in $\theta$. Next, we define $\hat{\hat{\Theta}} =(\hat\theta_{j})_{j\in J}$
by splitting each $\hat\theta_i$ at each point the time advances and retaining the order in $J \subseteq \mathbb{N}$. We can assume $J$ is a closed interval starting at 0. 
Now show (d-f) for each $\hat\theta_{j}$. We need one more observation: once \abs{cll} is set to 1, it is never set back, as it can only be manipulated by statements generated during the translation and we only generate statements setting it to 0.
\begin{description}
\item[Base Case, $j = 0$] This means this is the first time \fid is scheduled. 
Up until the end of $\hat\theta_{0}$, no \abs{await} or \abs{duration} was executed. 
If a \abs{get} was executed, then it was resolved immediately, the translation is obviously an overapproximation of any read value.
By (1) we know that $\hat\theta_{0}(0)$ is a model for the invariant and precondition, by (3) and the above observation we know no statement \abs{cll := 1} was executed. 

Now, we need to show that the last state in $\hat\theta_{0}$ is reachable by $\mathsf{trans}(\statement_\xabs{m})$.  
If only sequence, branching, loops and assignments without side-effects occur, this is trivial as the semantics of these statements coincide with their translation in \ddl., so the second-to-last state is reachable. 
Due to the split, we only need to consider \abs{await}, \abs{duration}, \abs{get} and \abs{return} as the statement causing the last transition. In between method calls may occur. We observe that if the method precondition would be broken, then it would be at the call site, because the precondition only specifies the parameters. Thus, at each calling configuration $\hat\theta_{0}(k)$ we need to ensure that the passed values are a model for the precondition. Now, $\hat\theta_{0}(k)$ is reachable by $\mathsf{trans}$ applied to the statements above. If the precondition would not hold, that at this point the test $\neg\mathsf{pre}(...)$ would succeed and, thus, \abs{cll := 1} would be executed by the translated statement. However, by (3) and the above observation we know that this is not the case. Thus, all method preconditions hold in $\hat\theta_{0}$ and we have shown (d). 
We observe that no matter which statement is responsible for the split, it is explicitly checked that the invariant holds (we remind that $\mathsf{pr}$ has the form $I \wedge ...$).
If it would not hold, \abs{cll} would not be 0, which is guaranteed by (3).

\item[Step Case, $j > 0$] This means that either \fid is rescheduled or time advanced while it was active.
\begin{description}
\item[Case: Time Advance] I.e., there was a \abs{get} expression or a \abs{duration} expression.
The case of \abs{get} is trivial, as we show that the system remains safe forever when the statement is reached, thus, it does not matter how long the advance was. We must also show that the state after time advance corresponds to a state after the symbolic execution of the translation of the \abs{get}. 
For this it suffices that $\mathsf{havoc;}?I$ overapproximates all states where only the fields change and the invariant holds. The case for \abs{duration} is similar, except that we can overapproximate more closely, as we have an exact literal for the time advance. 
\item[Case: Rescheduled] I.e., there was an \abs{await} statement. We must show that the configuration after reactivation is abstracted by the state after the translation of \abs{await}. We can, by (2) assume the invariant and only the fields may have changed. We may assume the guard, as otherwise the rescheduling would not be possible. Thus, removing all fields to some state that only assume the guard and the invariant is sound.
\end{description}
For the last configuration of $\hat\theta_j$, the case is analogous to the base case.
\end{description}
\end{proof}

We can now prove the theorem itself, going from traces to runs.
\begin{quote}
If $\psi_\xabs{C.m}$ is \textbf{Prgm}-sound for all \abs{C.m}, then the proof obligation scheme of Def.~\ref{def:scheme} is sound for \textbf{Prgm} in the sense of Def.~\ref{def:soundscheme}.
\end{quote}
\begin{proof}
It suffices to show that (1) and (2) from the above lemma hold for every \fid and that this implies that whenever time advances and no process is active for an object, the invariant still holds.
To do so, we again apply an induction, this time on the length of the \emph{run}.
We show that the above property holds up to the $n$th transition.
\begin{description}
\item[Base Case, $n = 0$] In the initial configuration only the main block is executed in a special method of a special class. This class has only trivial specification and, in particular, no object invariant.
\item[Step Case, $n > 0$] We can assume that the run was safe for $n-1$ transitions and now the $n$th transition takes place.
\begin{description}
\item[Case: Timed transition] I.e., the run so far ends in a configuration 
\[\mathsf{clock}(t_{n-1}) \mathit{cn}_{n-1}\]
the next transition results in a configuration
\[\mathsf{clock}(t_{n-1}+t') \mathit{adv}(\mathit{cn}_{n-1},t') = \mathsf{clock}(t_n) \mathit{cn}_{n} \]
We need to show that for all $t$ with $t_{n-1} < t \leq t_n$, all objects are safe.
For each object the on-going suspension subtrace starts in a configuration before $\mathsf{clock}(t_{n-1}) \mathit{cn}_{n-1}$. Thus, it starts in a safe configuration. As the post-region is a sound approximation of the suspension subtrace, each interpolated state is safe if following the dynamics from a safe configuration remains safe inside the sound post-region. If the object in question is suspended during time advance, this is exactly expressed by the post-condition and the check in the \abs{await} translation. As all proof obligations are closed and our proof obligations are only valid if additionally suspension at blocking statements with \abs{get} is safe forever, this is indeed the case and the reached state is safe. 

\item[Case: Discrete transition with activation (\rulename{i},\rulename{3})] 
The only change is that a process is moved from the active position to the process pool.
As the store is not modified, if $\mathit{cn}_{n-1}$ is safe, so is $\mathit{cn}_{n}$. 

\item[Case: Discrete transition with suspension (\rulename{i},\rulename{2})] 
The only change is that a process is moved from the active position to the process pool for some object $o$.
We can ignore all other objects, as their state does not change and safety is preserved. 
As the store is not modified, we must show that in $\mathit{cn}_{n-1}$ is safe for $o$.
We use the validity of the proof obligation: as we have shown each process establishes the invariant if it can assume it in lemma 1. The process that is suspending now started in some $\mathit{cn}_{m}$ with $m < n-1$. Thus, it was safe then. As the proof obligation is valid, we can assume by lemma 1 that it reestablishes the invariant.
\item[Case: Internal discrete transition] Rule \rulename{i} triggers and an object makes a step by reducing a statement of an active process, i.e., the internal rule is any rule except \rulename{2} and \rulename{3} (and the analogous rule to start a process for the first time~\cite{BjorkBJST13}). This case is not relevant for the safety property as internal actions are completely handled by lemma~\ref{lem:compose}.
\end{description}
\end{description}
\end{proof}

\subsection{Proof for Theorem~\ref{lem:uniform}}
\begin{quote}
Given a class \abs{C} and a method \abs{m}, let $\mathsf{calls}_{\xabs{C.m}}$ be the set of methods that are guaranteed to be called within \abs{C.m}.

The post-region generator $\Psi^\mathsf{local}$ that maps every method \abs{C.m} to 
\[
\bigwedge_{\methodname \in \mathsf{gcall}(\xabs{C.m})} \widetilde{\neg}\mathit{ttrig}_\methodname
\]
and every \abs{await} identifier $i$ to
\[
\bigwedge_{\methodname \in \mathsf{gcall}(\xabs{i})} \widetilde{\neg}\mathit{ttrig}_\methodname
\]
is sound.
\end{quote}
\begin{proof}
We show the case for $\mathsf{gcall}(\xabs{C.m})$, the one for program-point identifier is analogous.
We must show that every state in every suspension-subtrace of $x$ is a model for 
\[
\mathit{post} =  \bigwedge_{\methodname \in \mathsf{gcall}(\xabs{C.m})} \widetilde{\neg}\mathit{ttrig}_\methodname
\]
Let $o$ be an object of class \classname executing method \methodname.
First, we observe that when $o$ terminates, the configuration has the following form, because every member of $\mathsf{gcall}(\xabs{C.m})$ is guaranteed to be called,
but no suspension occured that could remove the messages\footnote{This is a standard property of CFGs, indeed we do not require the CFG-based analysis to generate $\mathsf{gcall}$ and can use the semantical characterization we just stated.}:
\[
\mathsf{clock}(i)~
(o,\rho,ODE,f,\bot,q)~
\mathsf{msg}(o,\methodname_1,\xabs{e}^1_1\dots\xabs{e}^1_n,\mathit{fid}_1)~
\dots~
\mathsf{msg}(o,\methodname_j,\xabs{e}^j_1\dots\xabs{e}^j_m,\mathit{fid}_j)~
\mathit{cn}
\]
Before time can advance, and thus the generation of the suspension-subtrace can begin, the next steps will be using rule \rulename{8} to turn all messages into new process of $o$.
Afterwards, the configuration has the following form:
\[
\mathsf{clock}(i)~
\big(o,\rho,ODE,f,\bot,q\cdot(\tau_1,\mathit{fid}_1,\xabs{s}_1)\dots\cdot(\tau_j,\mathit{fid}_j,\xabs{s}_j)\big)~
\mathit{cn}
\]
Where $\xabs{s}_k$ is the method body of the $k$th called method.
The next rule that is applied to the $o$ must be \rulename{ii}, as otherwise the suspension-subtrace would not exist, or there is no further operation.
In either case, let us assume that there would be a state at time $k>i$ in the suspension-subtrace that would not by a model for $\mathit{post}$.
In this case, at least one of the guards of the called method would trigger and reschedule its process identified by $\mathit{fid}_j$.
But if this is the case, then the state of $k$ cannot be part of the suspension-subtrace. \qedhere
\end{proof}

\subsection{Proof for Theorem~\ref{lem:control}}
\begin{quote}
The post-region generator $\Psi^\mathsf{ctrl}$, which maps every method \abs{C.m}, and every identifier $i$ of an \abs{await} with \abs{C.m},
to the following formula,
\[
\bigwedge_{\methodname \in \mathsf{Ctrl}_{\xabs{C}} \cup\mathsf{CM}} \widetilde{\neg}\mathit{ttrig}_{\methodname}
\]
is sound.
\end{quote}
\begin{proof}
Let $x$ be a program-point identifier or a method. Let $o$ be an object of the class in question, execution \abs{C.m}.
We must show that every state in every suspension-subtrace of $x$ is a model for 
\[
\mathit{post} = \bigwedge_{\methodname \in \mathsf{Ctrl}_{\xabs{C}}} \widetilde{\neg}\mathit{ttrig}_{\methodname}
\]
First, we observe that every suspension-subtrace starts at an index $i$ of some where no process is active.
Second, every controller is, by definition, always schedulable, i.e., there is a process for each controller executing the contoller method.
I.e., the configuration has the form 
\[
\mathsf{clock}(i)~\big(o,\rho,ODE,f,\bot,q\cdot(\tau_1,\mathit{fid}_1,\xabs{rs}_1)~\dots~\cdot(\tau_n,\mathit{fid}_n,\xabs{rs}_n)\big)~\mathit{cn}
\]
Where $\xabs{rs}_j = \xabs{await diff g}_j$ is the method body of the $j$th controller method.
The next rule that is applied to the $o$ must be \rulename{ii}, as otherwise the suspension-subtrace would not exist, or there is no further operation.
In either case, let us assume that there would be a state at time $k>i$ in the suspension-subtrace that would not by a model for $\mathit{post}$.
In this case, at least one of the guards of the processes would trigger and reschedule its process identified by $\mathit{fid}_j$.
But if this is the case, then the state of $k$ cannot be part of the suspension-subtrace. \qedhere
\end{proof}

\tableofcontents
\end{document}